\documentclass{article}

\usepackage{amsthm}
\usepackage{amssymb,amsmath,amsfonts}

\usepackage{xcolor,tikz}
\usepackage{subfigure}
\usepackage{url,xspace}

\usetikzlibrary{calc,shapes,arrows}
\usepackage{ifthen}
\usepackage{varioref}
\usepackage{enumerate}
\usepackage{xspace}
\usepackage{textcomp}

\newcommand{\todo}[1]{%
  { \renewcommand{\baselinestretch}{0.3} %
  \begin{tikzpicture}[remember picture]%
      \node [coordinate] (inText) {};%
  \end{tikzpicture} %
  \marginpar[{%
      \begin{tikzpicture}[remember picture]%
          \draw node[draw=red, anchor=north west, color=red, text width = 1.15cm, font=\scriptsize, inner sep=0.10cm,yshift=-0.15cm] (inNote)%
                   {#1};%
      \end{tikzpicture}%
      \begin{tikzpicture}[remember picture, overlay] %
    \draw[draw = red]%
        ([yshift=-0.10cm] inText)%
        -| ([xshift=0.05cm] inNote.south east) %
        -| ([xshift=0.05cm] inNote.east) %
        -| (inNote.east);%
      \end{tikzpicture}%
  }]{ %
      \begin{tikzpicture}[remember picture]%
          \draw node[draw=red, anchor=north west, color=red, text width = 1.15cm, font=\scriptsize, inner sep=0.10cm,yshift=-0.15cm] (inNote)%
                   {#1};%
      \end{tikzpicture}%
      \begin{tikzpicture}[remember picture, overlay] %
    \draw[draw = red]%
        ([yshift=-0.10cm] inText)%
        -| ([xshift=-0.05cm] inNote.south west)%
        -| ([xshift=-0.05cm] inNote.west)%
        -| (inNote.west);%
      \end{tikzpicture}%
  }}}%
\renewcommand{\todo}[1]{}

\newtheorem{theorem}{Theorem}[section]
\newtheorem{lemma}[theorem]{Lemma}
\newtheorem{corollary}[theorem]{Corollary}
\newtheorem{proposition}[theorem]{Proposition}

\newtheorem{definition}[theorem]{Definition}

\newtheorem{claim}[theorem]{Claim}

\newcommand{\theoremlike}[2]{\par\medskip\penalty-250
{{%
\textsc{
#2 \ref{#1}.}}}\it}

\newcommand{\thmhelperpre}[2]{\theoremlike{#1}{#2}}
\newcommand{\thmhelperpost}{\par\medskip}

\newenvironment{reftheorem}[1]{\thmhelperpre{#1}{Theorem}}{\thmhelperpost}
\newenvironment{reflemma}[1]{\thmhelperpre{#1}{Lemma}}{\thmhelperpost}

\newenvironment{refcorollary}[1]{\thmhelperpre{#1}{Corollary}}{\thmhelperpost}
\newenvironment{refproposition}[1]{\thmhelperpre{#1}{Proposition}}{\thmhelperpost }

\newcommand{\Nset}{\mathbb{N}}
\newcommand{\Nseto}{\Nset_0}

\newcommand{\Rset}{\mathbb{R}}
\newcommand{\Rsetp}{\mathbb{R}_{>0}}
\newcommand{\Rsetpo}{\mathbb{R}_{\ge 0}}

\newcommand{\calP}{\mathcal{P}}
\newcommand{\calX}{\mathcal{X}}
\newcommand{\calY}{\mathcal{Y}}
\newcommand{\calB}{\mathcal{B}}
\newcommand{\calA}{\mathcal{A}}

\newcommand{\calC}{\mathcal{C}}

\newcommand{\dist}{\mathcal{D}}
\newcommand{\G}{\mathcal{G}}

\newcommand{\prob}{\mathbf{P}}

\newcommand{\D}{\mathbf{D}}
\newcommand{\distribution}{\alpha}

\newcommand{\sigmafield}{\mathcal{F}}
\newcommand{\initmeasure}{\mu}

\renewcommand{\vec}[1]{\mathbf{#1}}

\newcommand{\fr}{\mathit{frac}}  %
\newcommand{\de}[1]{\mathit{d#1}}  %

\newcommand{\delay}{\mathfrak{D}}
\newcommand{\smp}{\mathcal{M}}
\newcommand{\setofruns}{\mathcal{R}}
\newcommand{\gssmc}{\Phi}
\newcommand{\ta}{\mathcal{A}}
\newcommand{\product}{\ensuremath{{\smp \times \ta}}\xspace}
\newcommand{\borel}{\mathfrak{B}}
\newcommand{\kernel}{P}
\newcommand{\probm}{\calP}

\newcommand{\successor}{\ta}

\newcommand{\pr}{\xi} %

\newcommand{\region}{\sim} %
\newcommand{\sizereg}{|V|}

\newcommand{\regiongraph}{G} %

\newcommand{\productfield}{\otimes}
\newcommand{\Productfield}{\bigotimes}

\newcommand{\smallmeasure}{\kappa}

\newcommand{\vecv}{\vec{v}}

\newcommand{\ntime}[1]{T^{#1}}

\newcommand{\nstate}[2]{1^{#1}_{#2}}
\newcommand{\maxb}{B_{\max}}
\newcommand{\densb}{c_\delay}
\newcommand{\pmin}{p_{\min}}

\newcommand{\dme}{\mathbf{d}}
\newcommand{\cme}{\mathbf{c}}

\newcommand{\trans}[3]{#1{}\mathchoice%
    {\stackrel{#2}{\rightarrow}}
    {\mathop {\smash\rightarrow}\limits^{\vrule width 0pt height 0pt
                                                depth 4pt\smash{#2}}}
    {\stackrel{#2}{\rightarrow}}
    {\stackrel{#2}{\rightarrow}}
{}#3}

\newcommand{\ltran}[1]{{}\mathchoice%
    {\stackrel{#1}{\longrightarrow}}
    {\mathop {\smash\longrightarrow}\limits^{\vrule width 0pt height 0pt
                                                depth 4pt\smash{#1}}}
    {\stackrel{#1}{\longrightarrow}}
    {\stackrel{#1}{\longrightarrow}}
{}}

\tikzstyle{max}=[thick,draw,minimum size=1.4em,inner sep=0em]
\tikzstyle{min}=[diamond,thick,draw,minimum size=1.4em,%
    inner sep=0em]
\tikzstyle{ran}=[circle,thick,draw,minimum size=2.5em,%
    inner sep=0em]
\tikzstyle{mc}=[rounded corners,thick,draw,minimum size=1.4em,%
    inner sep=.5ex]
\tikzstyle{tran}=[thick,draw,->,>=stealth]
\tikzstyle{loop left}=[tran, to path={.. controls +(150:.5)
    and +(210:.5) .. (\tikztotarget) \tikztonodes}]
\tikzstyle{loop right}=[tran, to path={.. controls +(30:.5)
    and +(330:.5) .. (\tikztotarget) \tikztonodes}]
\tikzstyle{loop above}=[tran, to path={.. controls +(60:.5)
    and +(120:.5) .. (\tikztotarget) \tikztonodes}]
\tikzstyle{loop below}=[tran, to path={.. controls +(240:.5)
    and +(300:.5) .. (\tikztotarget) \tikztonodes}]

\definecolor{Shadow}{rgb}{0.729,0.337,0.827}

\begin{document}

\title{Measuring Performance of Continuous-Time Stochastic Processes 
using Timed Automata}

\author{Tom\'a\v{s} Br\'azdil 
 \and
 Jan Kr\v{c}\'al
 \and
 Jan K{\v{r}}et\'{\i}nsk\'y
 \and
 Anton\'{\i}n Ku\v{c}era
 \and
 Vojt\v{e}ch~{\v{R}}eh\'ak
}

\date{Faculty of Informatics, Masaryk University, \\
 Botanick\'{a} 68a, 60200 Brno,\\  Czech Republic\\
 \texttt{\{brazdil,\,krcal,\,jan.kretinsky,\,kucera,\,rehak\}\\
@fi.muni.cz}\\}

\maketitle

\thispagestyle{empty}

\widowpenalty=10000
\clubpenalty=10000
\sloppy

\begin{abstract}
We propose deterministic timed automata (DTA) as a model-independent language
for specifying performance and dependability measures over
continuous-time stochastic processes. Technically, these measures
are defined as limit frequencies of locations (control states) of a DTA
that observes computations of a given stochastic process.
Then, we study the properties of DTA measures over semi-Markov processes 
in greater detail. We show that DTA measures over semi-Markov processes
are well-defined with probability one, and there are only finitely many
values that can be assumed by these measures with positive
probability. We also give an algorithm which approximates
these values and the associated probabilities up to an arbitrarily
small given precision. Thus, we obtain a general and effective
framework for analysing DTA measures over semi-Markov processes.
\end{abstract}

\section{Introduction}
\label{sec:intro}
\newcounter{jj}

Continuous-time stochastic processes, such as continuous-time
Markov chains, semi-Markov processes, or generalized semi-Markov
processes \cite{Ross:book,BL:book,Norris:book,Matthes:GSMP}, 
have been widely used in practice to determine
performance and dependability characteristics of real-world
systems. The desired behaviour of such systems is specified
by various measures such as mean response time, throughput,
expected frequency of errors, etc. These measures are often 
formulated just semi-formally and chosen specifically
for the system under study in a somewhat ad hoc manner. 
One example of a rigorous and model-independent specification
language for performance and dependability properties is Continuous 
Stochastic Logic (CSL) 
\cite{ASSB:CTMC-model-checking-ACM-TCL,BHHK:CTMC-model-checking-IEEE-TSE}
which allows to specify both steady state 
and transient measures over the underlying stochastic
process. The syntax and semantics of CSL is inspired by the well-known
non-probabilistic logic CTL~\cite{EC:CTL-SCP}. 
The syntax of CSL defines \emph{state}
and \emph{path} formulae, interpreted over the states and runs
of a given stochastic process $\smp$. In particular, there are two 
probabilistic operators, $\mathcal{P}_{{\Join}\varrho}(\cdot)$ and 
$\mathcal{S}_{{\Join}\varrho}(\cdot)$, which refer to the transient and
steady state behaviour of~$\smp$, respectively. Here
$\Join$ is a numerical comparison (such as $\leq$) and $\varrho \in [0,1]$
is a rational constant. If $\varphi$ is a path 
formula\footnote{In CSL, $\varphi$ can be of the form $\mathcal{X}_I \Phi$ or 
$\Phi_1 \mathcal{U}_I \Phi_2$ where $\Phi,\Phi_1,\Phi_2$ are state formulae, and
$\mathcal{X}_I,\mathcal{U}_I$ are the modal connectives of CTL
parametrized by an interval~$I$.
Boolean connectives can be used to combine just state formulae.}
(which is either valid or invalid for every run of~$\smp$), 
then $\mathcal{P}_{{\geq}0.7}(\varphi)$ is a state formula which
says ``the probability of all runs satisfying $\varphi$ is at least~$0.7$''. 
If $\Phi$ is a state formula, i.e., $\Phi$ is either valid or invalid 
in every state, then  $\mathcal{S}_{{\geq}0.5}(\Phi)$ is also a state 
formula which says ``the $\pi$-weighted sum over all states where 
$\Phi$ holds is at least $0.5$''. Here $\pi$ is the steady-state
distribution of~$\smp$. 
The logic CSL can express quite complicated properties 
and the corresponding model-checking problem over continuous-time
Markov chains is decidable. However, there are also several disadvantages.    
\begin{itemize}
\item[(a)] The semantics of steady state probabilistic operator
   $\mathcal{S}_{{\Join}\varrho}(\cdot)$ assumes the existence of invariant
   distribution which is not guaranteed to exist for all types of
   stochastic processes with continuous time (the existing works mainly
   consider CSL as a specification language for ergodic continuous-time
   Markov chains).  
\item[(b)] In CSL formulae, all measures are explicitly quantified, and
   the model-checking algorithm just verifies constraints over these
   measures. Alternatively, we might wish to \emph{compute} certain
   measures up to a given precision. 
\end{itemize}

In this paper, we propose \emph{deterministic timed automata (DTA)} 
\cite{AD:Timed-Automata} as a
model-independent specification language for performance and dependability
measures of continuous-time stochastic processes. The ``language''
of DTA can be interpreted over \emph{arbitrary} stochastic processes 
that generate timed words, and their expressive power appears 
sufficiently rich to capture many interesting run-time properties
(although we do not 
relate the expressiveness of CSL and DTA formally, they are surely 
incomparable because of different ``nature'' of the two formalisms).
Roughly speaking, a DTA $\calA$ ``observes'' runs of a given stochastic 
process $\smp$ and ``remembers'' certain information in its control 
states (which are called \emph{locations}).
Since $\calA$ is deterministic, for every run $\sigma$ of $\smp$ there 
is a unique computation $\calA(\sigma)$ of $\calA$, which determines
a unique tuple of ``frequencies'' of visits to the individual 
locations of~$\calA$ along~$\sigma$. These frequencies are the values 
of ``performance measures'' defined by $\calA$ (in fact, we consider
\emph{discrete} and \emph{timed} frequencies which are based on the 
same concept but defined somewhat differently). 

Let us explain the idea in more detail. Consider some 
stochastic process $\smp$ whose computations (or runs)
are infinite sequences of the form $\sigma = s_0\,t_0\,s_1\,t_1\cdots$ where
all $s_i$ are ``states'' and $t_i$ is the time spent by performing
the transition from $s_i$ to $s_{i+1}$. Also assume a suitable 
probability space defined over the runs of~$\smp$. Let $\Sigma$ by a finite 
alphabet and $L$ a labelling which assigns a unique letter
$L(s) \in \Sigma$ to every state $s$ of $\smp$.    
Intuitively, the letters of $\Sigma$ correspond to collections of
predicates that are valid in a given state.
Thus, every run $\sigma = s_0\,t_0\,s_1\,t_1\cdots$
of $\smp$ determines a unique \emph{timed word} 
$w_\sigma = L(s_0)\,t_0\,L(s_1)\,t_1\cdots$ over~$\Sigma$. 

A DTA over $\Sigma$ is a finite-state automaton 
$\calA$ equipped with finitely many internal clocks. Each control 
state (or location) $q$ of $\calA$ has finitely many outgoing 
edges $q \ltran{} q'$ labeled by triples $(a,g,X)$, where
$a \in \Sigma$, $g$ is a ``guard'' (a constraint on the current 
clock values), and $X$ is a subset of clocks that are reset to zero
after performing the edge. A configuration of $\calA$ is 
a pair $(q,\nu)$, where $q$ and $\nu$ are the current location and 
the current clock valuation, respectively. Every timed word
$w = c_0\,c_1\,c_2\,c_3\cdots$ over $\Sigma$ (where $c_i \in \Sigma$ 
iff $i$ is even) then determines a unique
run $\calA(w) = (q_0,\nu_0)\,(q_1,\nu_1)\,(q_2,\nu_2)\cdots$ of $\calA$ where
$q_0$ is an \emph{initial} location, $\nu_0$ assigns zero
to every clock, and $(q_{i+1},\nu_{i+1})$ is obtained from $(q_{i},\nu_{i})$
either by performing the only enabled edge $q_i \ltran{} q_{i+1}$
labeled by $(c_i,g,X)$ if $i$ is even, or by simultaneously increasing
all clocks by $c_i$ if $i$~is odd. 

As a simple example, consider the
following DTA~$\hat{\calA}$ over the alphabet $\{a\}$ with one clock $x$
and the initial location~$q_0$:
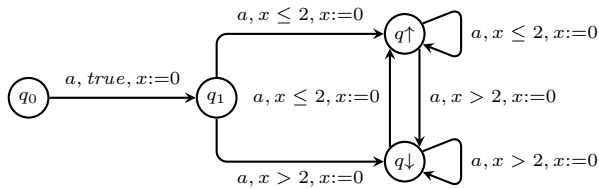
\begin{center}
\begin{tikzpicture}[x=2.5cm,y=1.7cm,font=\scriptsize]
\node (q0) at (0,0)   [ran,minimum size=1.5em] {$q_0$};
\node (q1) at (1,0)   [ran,minimum size=1.5em] {$q_1$};
\node (qu) at (2,.5)  [ran,minimum size=1.5em] {$q{\uparrow}$};
\node (qd) at (2,-.5) [ran,minimum size=1.5em] {$q{\downarrow}$};
\draw [tran] (q0) --
    node[above] {$a,\mathit{true},x{:=}0$} (q1);
\draw [tran, rounded corners] (q1) -- +(0,.5) --
    node[above] {$a,x\leq 2,x{:=}0$} (qu);
\draw [tran, rounded corners] (q1) -- +(0,-.5) --
    node[below] {$a,x> 2,x{:=}0$} (qd);
\draw [tran] (qu.south east) -- 
    node[right] {$a,x> 2,x{:=}0$} (qd.north east);
\draw [tran] (qd.north west) -- 
    node[left] {$a,x \leq 2,x{:=}0$} (qu.south west);
\draw [tran, rounded corners] (qu.30) -- (2.3,.7) -- 
    node[right] {$a,x \leq 2,x{:=}0$} (2.3,.3) -- (qu.330);
\draw [tran, rounded corners] (qd.30) -- (2.3,-.3) -- 
    node[right] {$a,x> 2,x{:=}0$} (2.3,-.7) -- (qd.330);
\end{tikzpicture}
\end{center}
Intuitively, $\hat{\calA}$ observes time stamps in a given timed 
word and enters 
either $q{\uparrow}$  or $q{\downarrow}$ depending on whether a given 
stamp is bounded by~$2$ or not, respectively. For example, a word 
$w = a\ 0.2\ a\ 2.4\ a\ 2.1\cdots$ determines the
run $\hat{\calA}(w) = 
(q_0,0)\,(q_1,0)\,(q_1,0.2)\,(q{\uparrow},0)\,(q{\uparrow},2.4)\,
(q{\downarrow},0)\,(q{\downarrow},2.1)\cdots$

Let $w = a_0\,t_0\,a_1\,t_1\cdots$ be a timed word over $\Sigma$ and $q$
a location of $\calA$. For every $i \in \Nset_0$, let $T^i(w)$ be
the stamp $t_i$ of $w$, and $Q^i(w)$ the location of
$\calA$ entered after reading the finite prefix
$a_0\,t_0\cdots a_i$ of~$w$. Further, let $\nstate{i}{q}(w)$
be  either $1$ or $0$ depending on whether $Q^i(w) = q$ or not,
respectively. We define the \emph{discrete} 
and \emph{timed} frequency of visits to $q$ along $\calA(w)$, denoted
by $\dme^{\calA}_{q}(w)$ and $\cme^{\calA}_{q}(w)$, in the following way (the
`$\calA$' index is omitted when it is clear from the context):
\begin{eqnarray*}
  \dme^{\calA}_q(w) & = &
  \limsup_{n \rightarrow \infty} \frac{\sum_{i=1}^{n} \nstate{i}{q}(w)}{n}\\
  \cme^{\calA}_q(w) & = &
  \limsup_{n \rightarrow \infty} \frac{\sum_{i=1}^{n} T^i(w) \cdot 
     \nstate{i}{q}(w)}{\sum_{i=1}^n T^i(w)}
\end{eqnarray*}
Thus, every timed word $w$ determines the tuple 
\mbox{$\dme^{\calA}(w) = \left(\dme^{\calA}_q(w) \right)_{q\in Q}$} 
and the tuple $\cme^{\calA}(w) = \left(\cme^{\calA}_q(w) \right)_{q\in Q}$ of 
\emph{discrete} and \emph{timed $\calA$-measures}, respectively. 

DTA measures can encode various performance and dependability properties
of stochastic systems with continuous time. For example, consider
again the DTA $\hat{\calA}$ above and assume that 
all states of a given stochastic process $\smp$ are labeled with~$a$.
Then, the fraction 
\[
  \frac{\dme_{q{\uparrow}}(w_\sigma)}
       {\dme_{q{\uparrow}}(w_\sigma) + \dme_{q{\downarrow}}(w_\sigma)}
\] 
corresponds to the percentage of transitions of $\smp$ that are performed
within~$2$ seconds along a run~$\sigma$. If $\smp$ is an ergodic 
continuous-time Markov chain, then the above fraction 
takes the \emph{same} value for almost all runs $\sigma$ of~$\smp$. 
However, it makes sense to consider this fraction also for non-ergodic
processes. For example, we may be interested in the expected
value of $\dme_{q{\uparrow}}/(\dme_{q{\uparrow}} + \dme_{q{\downarrow}})$, or in 
the probability of all runs~$\sigma$ such that the fraction is 
at least $0.5$.
 
One general trouble with DTA measures is that $\dme^{\calA}_q(w)$ and $\cme^{\calA}_q(w)$
faithfully capture the frequency of visits to~$q$ along $w$ only if 
the limits
\[
  \lim_{n \rightarrow \infty} \frac{\sum_{i=1}^{n} \nstate{i}{q}(w)}{n}
\mbox{\quad and \quad}
\lim_{n \rightarrow \infty} \frac{\sum_{i=1}^{n} T^i(w) \cdot 
\nstate{i}{q}(w)}{\sum_{i=1}^n T^i(w)}
\] 
exist, in which case we say that $\dme^\calA$ and $\cme^\calA$
are \emph{well-defined} for~$w$, respectively.
So, one general question that should be answered when
analyzing the properties of DTA measures over a particular class 
of stochastic processes is whether $\dme^\calA$ and $\cme^\calA$
are well-defined for almost all runs.
If the answer is negative, we might either
try to re-design our DTA or accept the fact that the limit
frequency of the considered event simply does not exist (and stick to
$\limsup$).

In this paper, we study DTA measures over semi-Markov processes (SMPs).
An SMP is essentially a discrete-time Markov chain where each transition
is assigned (apart of its discrete probability) a \emph{delay density},
which defines the distribution of time needed to perform the transition. 
A computation (run) of an SMP $\smp$ is initiated in some state $s_0$, 
which is also chosen randomly according to a fixed initial 
distribution over the state space of $\smp$.
The next transition is selected according to the fixed transition 
probabilities, and the selected transition 
takes time chosen randomly according to the density associated to
the transition. Hence, each run of $\smp$ is an infinite sequence 
$s_0\,t_0\,s_1\,t_1\cdots$, where all $s_i$ are states of $\smp$ and 
$t_i$ are time stamps. The probability of (certain) subsets 
of runs in $\smp$ is measured in the standard way 
(see Section~\ref{sec:prelim}). 

\emph{The main contribution of this paper} are general
results about DTA measures over semi-Markov processes, which 
are valid for all SMPs where the employed density functions are 
bounded from zero on every closed subinterval (see 
Section~\ref{sec:prelim}). Under this assumption, we prove that
for every SMP $\smp$ and every DTA $\calA$ we have the following: 
\begin{enumerate}
\item[(1)] Both discrete and timed $\calA$-measures are well defined
  for almost all runs of $\smp$.
\item[(2)] Almost all runs of $\smp$ can be divided into
  finitely many pairwise 
  disjoint subsets $\setofruns_1,\ldots,\setofruns_k$ so that
  $\dme^{\calA}(w)$ takes the same value for almost all  
  \mbox{$w \in \setofruns_j$}, where $1 \leq j \leq k$. 
  The same result holds also for $\cme^{\calA}$. (Let us note
  that $k$ can be larger than~$1$ even if $\smp$ is strongly connected.)
\item[(3)] The observations behind the results of (1) and (2) can be
  used to compute the $k$ and effectively approximate
  the probability of all $\setofruns_j$ together with the associated values
  of discrete or timed $\calA$-measures up to an arbitrarily small given
  precision. More precisely, we show that these quantities 
  are expressible using the $m$-step transition kernel $\kernel^m$ of the
  product process $\product$ defined for $\smp$ and $\calA$ (see
  Section~\ref{SEC:PRODUCT-PROCESS}), and we give generic bounds 
  on the number of steps~$m$ that is sufficient to achieve 
  the required precision.
  The $m$-step transition kernel is defined by nested 
  integrals (see Section~\ref{sec:gssmc}) and can be approximated by numerical 
  methods (see, e.g., \cite{Kress:book,BLM:book}). 
  This makes the whole framework effective. The design of more efficient 
  algorithms as well as more detailed analysis applicable
  to concrete subclasses of SMP are left for future work.
\end{enumerate}

To get some intuition about potential applicability of our results
(and about the actual power of DTA which is hidden mainly in their 
ability to accumulate the total time of several 
transitions in internal clocks), let us start 
with a simple example. Consider the 
following itinerary for travelling between Brno and Prague:
\begin{center}
\begin{tabular}{r l l l l l}
 & Brno & Ku\v{r}im & Ti\v{s}nov & \v{C}\'{a}slav 
        & Prague\\[.2ex]\hline\\[-1.8ex]
\textit{arrival} 
   & & 1:15 & 2:30 & 3:30 & 4:50\\
\textit{departure}
   & 0:00 & 1:20 & 2:40 & 3:35
\end{tabular}
\end{center}
A traveller has to change a train at each of the three intermediate 
stops, and she needs at least 3~minutes to walk between the platforms.
Assume that all trains depart on time, but can be delayed. Further, assume
that travelling time between $X$ and $Y$ has density $f_{\textit{X-Y}}$. 
We wonder what is the chance that a traveller reaches Prague from Brno
without missing any train and at most 5~minutes after the scheduled
arrival. Answering this question ``by hand'' is not 
simple (though still possible). However, it is almost trivial to
rephrase this question in terms of DTA measures. The itinerary can be 
modeled by the following semi-Markov process, where the density $f$ 
is irrelevant and $\Sigma = \{\textit{B,K,T,\v{C},P}\}$. 
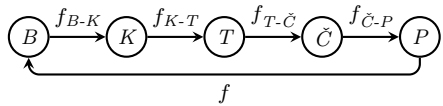
\begin{center}
\begin{tikzpicture}[x=1.3cm,y=1cm,font=\footnotesize]
\foreach \x/\ll in {0/B,1/K,2/T,3/\v{C},4/P}
  {
   \node (m\x) at (\x,0)   [ran,minimum size=1.5em] {$\textit{\ll}$};
  }
\foreach \x/\ll in {0/B-K,1/K-T,2/T-\v{C},3/\v{C}-P}
  {\setcounter{jj}{\x}\addtocounter{jj}{1}
   \draw [tran,rounded corners] (m\x) -- 
      node[above] {$f_{\textit{\ll}}$} (m\arabic{jj});
  }
\draw [tran,rounded corners] (m4) -- +(0,-.5) --
      node[below] {$f$} +(-4,-.5)-- (m0);
\end{tikzpicture} 
\end{center}
The property of ``reaching Prague from Brno
without missing any train and at most 5~minutes after the scheduled
arrival'' is encoded by the DTA~$\bar{\calA}$ of Figure~\ref{fig-TA-exa}. 
The automaton uses just one clock $x$ to measure the total elapsed time,
and the guards reflect the required timing constraints. Starting
in location $\mathit{init}$, the automaton eventually reaches 
either the location $p{\uparrow}$ or $p{\downarrow}$, which corresponds
to satisfaction or violation of the above property, and then it
is ``restarted''. Hence, we are interested in the relative frequency
of visits to $p{\uparrow}$ among the visits to 
$p{\uparrow}$ or $p{\downarrow}$. Using our results,
it follows that $\dme^{\calA}$ is well-defined
and takes the same value for almost all runs of~$\smp$.
Hence, the random variable
$\dme_{p{\uparrow}}/(\dme_{p{\uparrow}} + \dme_{p{\downarrow}})$ also takes the
same value with probability one, and this (unique) value is the quantity
of our interest.

Now imagine we wish to model and analyse the flow of passengers in
London metro at rush hours. The SMP states then correspond to
stations, transition probabilities encode the percentage of passengers
traveling in a given direction, and the densities encode the
distribution of travelling time.  A DTA can be used to monitor a
complex list of timing restrictions such as ``there is enough time to change a
train'', ``travelling between important stations does not take more than
30 minutes if one the given routes is used'', ``trains do not arrive more
than 2 minutes later than scheduled'', etc. For this we already 
need several internal clocks. Apart of some auxiliary locations, the
constructed DTA would also have special locations used to encode
satisfaction/violation of a given restriction (in the DTA $\bar{\calA}$ of
Figure~\ref{fig-TA-exa}, $(p,{\uparrow})$ and $(p,{\downarrow})$ are
such special locations). Using the results presented in this paper,
one may not only study the overall satisfaction of these restrictions,
but also estimate the impact of changes in the underlying model (for example,
if a given line becomes slower due to some repairs, one may evaluate
the decrease in various dependability measures without changing the
constructed DTA). 

\emph{Proof techniques.} For a given SMP $\smp$ and a given DTA $\calA$
we first construct their synchronized product $\smp \times \calA$, 
which is another stochastic process. In fact, it turns out
that \mbox{$\smp \times \calA$} is a discrete-time Markov chain with 
uncountable state-space. Then, we apply a variant of the standard 
region construction 
\cite{AD:Timed-Automata} and thus partition the state-space 
of $\smp \times \calA$ into finitely
many equivalence classes. At the very core of our paper there are 
several non-trivial observations about the structure of 
$\smp \times \calA$ and its region graph which establish a powerful
link to the well-developed ergodic theory of Markov chains with general
state-space (see, e.g., \cite{MT:book,RR:GSSMC-PS}). In this way, we obtain 
the results of items~(1) and~(2) mentioned above. 
Some additional work is required to analyze the algorithm presented
in Section~\ref{SEC:ALGORITHMS} (whose properties are summarized 
in item~(3) above).

\emph{Related work.} There is a vast literature on continuous-time Markov
chains, semi-Markov processes, or even more general stochastic models
such as generalized semi-Markov processes
(we refer to, e.g., \cite{Ross:book,BL:book,Norris:book,Matthes:GSMP}).
In the computer science context, most works on continuous-time stochastic
models concern model-checking against a given class of temporal properties 
\cite{ASSB:CTMC-model-checking-ACM-TCL,BHHK:CTMC-model-checking-IEEE-TSE}. 
The usefulness of CSL model-checking for dependability analysis
is advocated in \cite{HHK:CSL-dependability}. 
Timed automata \cite{AD:Timed-Automata} have been
originally used as a model of (non-stochastic) real-time systems.
Probabilistic semantics of timed automata is proposed in
\cite{BBBBG:One-clock-almost-sure,BBBM:One-clock-timed}.
The idea of using timed automata as a specification language 
for continuous-time stochastic processes is relatively
recent. In \cite{CHKM:CTMC-quant-timed}, the model-checking problem
for continuous-time Markov chains and linear-time properties represented
by timed automata is considered (the task is to dermine the probability
of all timed words that are accepted by a given timed automaton). 
A more general model of two-player
games over generalized semi-Markov processes with qualitative 
reachability objectives specified by deterministic timed automata
is studied in \cite{BKKKR:GSMP-games-TA}.

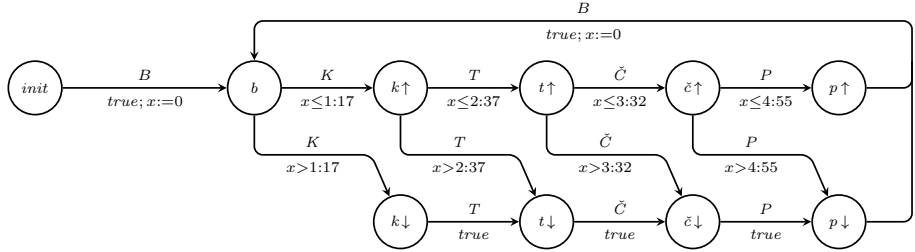
\begin{figure*}
\begin{center}
\resizebox{\linewidth}{!}{
\begin{tikzpicture}[x=2.4cm,y=2.2cm,font=\scriptsize]
\node (i) at (-1.5,0) [ran] {\textit{init}};
\node (mu0) at (0,0) [ran] {\textit{b}};
\foreach \x/\ll in {1/k,2/t,3/\v{c},4/p}
  {
   \node (mu\x) at (\x,0)   [ran] {$\textit{\ll}\,{\uparrow}$};
   \node (md\x) at (\x,-1)  [ran] {$\textit{\ll}\,{\downarrow}$};
  }
\draw [tran] (i) -- 
  node[below] {$\mathit{true}; x{:=}0$} 
  node[above] {$B$} (mu0);
\draw [tran,rounded corners] (mu4) -- +(.5,0) -- +(.5,.5) --
    node[below] {$\mathit{true}; x{:=}0$} node[above] {$B$}
    +(-4,.5) -- (mu0);
\draw [tran,rounded corners,-] (md4) -- +(.5,0) -- +(.5,1.2);
\foreach \x/\t/\ll in {0/1{:}17/K,1/2{:}37/T,2/3{:}32/\v{C},3/4{:}55/P}
  {\setcounter{jj}{\x}\addtocounter{jj}{1}
   \draw [tran,rounded corners] (mu\x) -- 
      node[below] {$x{\leq}\t$} 
      node[above] {\textit{\ll}} (mu\arabic{jj});
   \draw [tran,rounded corners] (mu\x) -- +(0,-.5) --
      node[below] {$x{>}\t$} 
      node[above] {\textit{\ll}} +(.8,-.5)-- (md\arabic{jj});
  }
\foreach \x/\ll in {1/T,2/\v{C},3/P}
  {\setcounter{jj}{\x}\addtocounter{jj}{1}
   \draw [tran,rounded corners] (md\x) -- 
      node[below] {$\mathit{true}$} 
      node[above] {\textit{\ll}} (md\arabic{jj});
  }
\end{tikzpicture}
}
\end{center}
\caption{A deterministic timed automaton $\bar{\calA}$.}
\label{fig-TA-exa}
\end{figure*}

\section{Preliminaries}
\label{sec:prelim}

In this paper, the sets of all positive integers, non-negative integers,
real numbers, positive real numbers, and non-negative real numbers are
denoted by $\Nset$, $\Nseto$, $\Rset$, $\Rsetp$, and $\Rsetpo$,
respectively.

Let $A$ be a finite or countably infinite set. A \emph{discrete
  probability distribution} on $A$ is a function $\distribution : A
\rightarrow \Rsetpo$ such that \mbox{$\sum_{a \in A} \distribution(a)
  = 1$}. We say that $\distribution$ is \emph{rational} if
$\distribution(a)$ is rational for every $a \in A$.
The set of all distributions on $A$ is denoted by $\dist(A)$.  A
\emph{$\sigma$-field} over a set $\Omega$ is a set $\sigmafield
\subseteq 2^{\Omega}$ that includes $\Omega$ and is closed under
complement and countable union.  A \emph{measurable space} is a pair
$(\Omega,\sigmafield)$ where $\Omega$ is a set called \emph{sample
  space} and $\sigmafield$ is a $\sigma$-field over $\Omega$ whose
elements are called \emph{measurable sets}.  A \emph{probability
  measure} over a measurable space $(\Omega,\sigmafield)$ is a
function \mbox{$\probm : \sigmafield \rightarrow \Rsetpo$} such that, for
each countable collection $\{X_i\}_{i\in I}$ of pairwise disjoint
elements of $\sigmafield$, $\probm(\bigcup_{i\in I} X_i) = \sum_{i\in
  I} \probm(X_i)$, and moreover $\probm(\Omega)=1$. A
\emph{probability space} is a triple $(\Omega,\sigmafield,\probm)$,
where $(\Omega,\sigmafield)$ is a measurable space and $\probm$ is a
probability measure over $(\Omega,\sigmafield)$.
We say that a property $A \subseteq \Omega$ holds \emph{for almost
  all} elements of a measurable set $Y$ if $\probm(Y) > 0$, $A \cap Y
\in \sigmafield$, and $\probm(A \mid Y) = 1$.

All of the integrals used in this paper should be understood as
Lebesgue integrals, although we use Riemann-like notation when
appropriate.
\subsection{Semi-Markov processes}
A semi-Markov process (see, e.g., \cite{Ross:book}) can be seen as
discrete-time Markov chains where each transition is equipped with a
density function specifying the distribution of time needed to perform
the transition.  Formally, let $\delay$ be a set of delay densities,
i.e., measurable functions $f : \Rset \rightarrow \Rsetpo$ satisfying
$\int_{0}^{\infty} f(t)\, \de{t} = 1$ where $f(t) = 0$ for every
$t<0$. Moreover, for technical reasons, we assume that each 
$f\in \delay$ satisfies the following: There is an interval $I$ either of
the form $[\ell,u]$ with $\ell,u \in \Nseto, \ell < u$, or 
$[\ell,\infty)$ with $\ell\in\Nseto$, such that
\begin{itemize}
 \item for all $t\in\Rset\setminus I$ we have that $f(t)=0$,
 \item for all $[c,d]\subseteq I$ there is $b>0$ such that for all 
   $t\in[c,d]$ we have that $f(t)\geq b$.
\end{itemize}
The assumption that $\ell,u$ are natural numbers is adopted only for
the sake of simplicity. Our results can easily be generalized to 
the setting where $I$ is an interval with rational bounds or even
a finite union of such intervals. 
\begin{definition}
  A \emph{semi-Markov process (SMP)} is a tuple $\smp = (S,\prob, \D,
  \distribution_0)$, where $S$ is a finite set of states, $\prob : S
  \rightarrow \dist(S)$ is a transition probability function, 
  \mbox{$\D : S \times S \rightarrow \delay$} is a delay function which to each
  transition assigns its delay density, and $\distribution_0 \in \dist(S)$
  is an initial distribution.
\end{definition}
A computation (run) of a SMP $\smp$ is initiated in some state $s_0$, 
which is chosen randomly according to $\distribution_0$. In the 
current state $s_i$, the next state $s_{i+1}$ is selected randomly according 
to the distribution $\prob(s_i)$, and the selected transition $(s_i,s_{i+1})$
takes a random time~$t_i$ chosen according to the density $\D(s_i,s_{i+1})$.
Hence, each run of $\smp$ is an infinite timed word 
$s_0\,t_0\,s_1\,t_1 \cdots$, where $s_i \in S$ and 
$t_i \in \Rsetpo$ for all $i \in \Nseto$. We use $\setofruns_{\smp}$
to denote the set of all runs of $\smp$.

Now we define a probability space 
$(\setofruns_{\smp},\sigmafield_{\smp},\probm_{\smp})$
over the runs of~$\smp$ (we often omit the index $\smp$ if it is clear
from the context). A \emph{template} is a finite sequence of the form 
$B = s_0\,I_0\,s_1\,I_1\cdots s_{n+1}$ such that
$n \geq 0$ and $I_i$ is an interval in~$\Rsetpo$
for every $0 \leq i \leq n$. Each such~$B$ determines the
corresponding \emph{cylinder} $\setofruns(B) \subseteq \setofruns$
consisting of all runs of the form $\hat{s}_0\,t_0\,\hat{s}_1\,t_1\cdots$, where
$\hat{s}_i = s_i$ for all \mbox{$0 \leq i \leq n{+}1$}, and $t_i \in I_i$ for
all $0 \leq i \leq n$. The
$\sigma$-field $\sigmafield$ is the Borel $\sigma$-field generated by all
cylinders. For each template $B = s_0\,I_0\,s_1\,I_1\cdots s_{n+1}$, let
$p_i = \prob(s_i)(s_{i+1})$ and $f_i = \D(s_i,s_{i+1})$ for all 
$0 \leq i \leq n$. The probability $\probm(\setofruns(B))$ is 
defined as follows:
\[
  \distribution_0(s_0) \cdot 
  \prod_{i=0}^{n} p_i \cdot \int_{t_i \in I_i} f_i(t_i)\, dt_i
\]
Then, $\probm$ is extended to $\sigmafield$ (in the unique way) by applying 
the extension theorem (see, e.g., \cite{Billingsley:book}).

\subsection{Deterministic timed automata}

Let $\calX$ be a finite set of \emph{clocks}.
A \emph{valuation} is a function $\nu: \calX \rightarrow \Rsetpo$.  
For every valuation $\nu$ and every subset 
$X \subseteq \calX$ of clocks, we use $\nu[X:=\vec{0}]$ to denote the unique
valuation such that $\nu[X:=\vec{0}](x)$ is equal either to $0$ or
$\nu(x)$, depending on whether $x \in X$ or not, respectively.
Further, for every valuation $\nu$ and every $\delta \in \Rsetpo$, the
symbol $\nu+\delta$ denotes the unique valuation such that $(\nu+\delta)(x)
= \nu(x) + \delta$ for all $x \in \calX$.  Sometimes we assume an implicite
linear ordering on clocks and slightly abuse our notation by 
identifying a valuation $\nu$ with the associated vector of reals.

A \emph{clock constraint} (or \emph{guard}) is a finite conjunction of
basic constraints of the form $x \bowtie c$, where $x \in \calX$,
\mbox{${\bowtie} \in \{{<},{\leq},{>},{\geq}\}$}, and $c \in
\Nseto$. For every valuation $\nu$ and every clock constraint $g$ we
have that $\nu$ either does or does not satisfy~$g$, written $\nu
\models g$ or $\nu \not\models g$, respectively (the satisfaction
relation is defined in the expected way). Sometimes we identify 
a guard~$g$ with the set of all valuations that satisfy $g$ 
and write, e.g., $g \cap g'$. The set of all guards over~$\calX$ 
is denoted by~$\calB(\calX)$.

\begin{definition}
  A \emph{deterministic timed automaton (DTA)} is a tuple $\ta = (Q,
  \Sigma, \calX, {\longrightarrow}, q_0)$, where $Q$ is a nonempty finite
  set of \emph{locations}, $\Sigma$ is a finite \emph{alphabet}, $\calX$ is
  a finite set of \emph{clocks}, $q_0 \in Q$ is an \emph{initial location},
  and
  ${\longrightarrow} \subseteq Q \times \Sigma \times \calB(\calX) \times
  2^\calX \times Q$ is an \emph{edge relation} such that for all $q \in Q$
  and $a \in \Sigma$ we have the following:
  \begin{enumerate}
  \item the guards are deterministic, i.e., for all edges of the form $(q,
    a, g_1, X_1,q_1)$ and $(q, a, g_2, X_2, q_2)$ such that $g_1 \cap g_2
    \neq \emptyset$ we have that $g_1 = g_2$, $X_1=X_2$, and $q_1=q_2$;
  \item the guards are total, i.e., for all $q \in Q$, $a \in \Sigma$, and
    every valuation $\nu$ there is an edge $(q,a,g,X,q')$ such that $\nu
    \models g$.
  \end{enumerate}
\end{definition}

\noindent
A \emph{configuration} of $\ta$ is a pair $(q,\nu)$, where $q \in Q$ and
$\nu$ is a valuation.  An \emph{infinite timed word} over $\Sigma$ 
is an infinite sequence $w = c_0\,c_1\,c_2\,c_3 \cdots$, where 
$c_i \in \Sigma$ when~$i$ is even, and $c_i \in \Rsetpo$ when $i$ is odd. 
The \emph{run} of $\ta$ on $w$ is the unique infinite sequence of 
configurations 
$\ta(w) = (q_0,\nu_0)\, (q_1,\nu_1) \cdots$ such that $q_0$ is the 
initial location of $\ta$, $\nu_0(x) = 0$ for all $x \in \calX$, 
and for each $i \in \Nseto$ we have that
\begin{itemize}
\item if $c_i$ is a time stamp, then $q_{i+1} = q_i$ and
  $\nu_{i+1} = \nu_i + c_i$;
\item if $c_i$ is a letter of $\Sigma$, then there is a unique
  edge $(q_i,c_i,g,X,q)$ such that $\nu_i \models g$, and we require that
  $q_{i+1} = q$ and $\nu_{i+1} = \nu_{i}[X := \vec{0}]$.
\end{itemize}
Notice that we do not define any acceptance condition for DTA.
Instead, we understand DTA as finite-state \emph{observers} that
analyze timed words and report about certain events by entering 
designated locations. The ``frequency'' of these events is formally
captured by the quantities $d_q$ and $c_q$ defined below.

Let $\ta = (Q, \Sigma, \calX, {\longrightarrow}, q_0)$ be a DTA,
$q \in Q$ some location, and $w = a_0\,t_0\,a_1\,t_1\cdots$ a timed 
word over $\Sigma$. For every $i \in \Nset_0$, let $T^i(w)$ be
the stamp $t_i$ of $w$, and $Q^i(w)$ the unique location of
$\calA$ entered after reading the finite prefix
$a_0\,t_0\cdots a_i$ of~$w$. Further, let $\nstate{i}{q}(w)$
be  either $1$ or $0$ depending on whether $Q^i(w) = q$ or not,
respectively. The \emph{discrete} and \emph{timed} frequency 
of visits to $q$ along $\calA(w)$, denoted by $\dme^{\calA}_q(w)$ and 
$\cme^{\calA}_q(w)$, are defined in the following way (if $\calA$ is
clear, it is omitted):
\begin{eqnarray*}
  \dme^{\calA}_q(w) & = &
  \limsup_{n \rightarrow \infty} \frac{\sum_{i=1}^{n} \nstate{i}{q}(w)}{n}\\
  \cme^{\calA}_q(w) & = &
  \limsup_{n \rightarrow \infty} \frac{\sum_{i=1}^{n} T^i(w) \cdot 
     \nstate{i}{q}(w)}{\sum_{i=1}^n T^i(w)}
\end{eqnarray*}
Hence, every timed word $w$ determines the tuple 
$\dme^{\calA} = \left(\dme^{\calA}_q(w) \right)_{q\in Q}$ 
and the tuple $\cme^{\calA} = \left(\cme^{\calA}_q(w) \right)_{q\in Q}$ of 
\emph{discrete} and \emph{timed $\calA$-measures}, respectively. 
The $\calA$-measures were defined using $\limsup$, because the 
corresponding limits may not exist in general. If 
$\lim_{n \rightarrow \infty} \sum_{i=1}^{n} \nstate{i}{q}(w)/n$
exists for all~$q \in Q$, we say  that
$\dme^{\calA}$ is \emph{well-defined} for~$w$. Similarly, if
$\lim_{n \rightarrow \infty} (\sum_{i=1}^{n} T^i(w) \cdot 
     \nstate{i}{q}(w))/(\sum_{i=1}^n T^i(w))$ exists for all~$q$,
we say that $\cme^{\calA}$ is \emph{well-defined} for~$w$.

As we already noted in Section~\ref{sec:intro}, a DTA $\calA$ can be used 
to observe runs in a given SMP~$\smp$ after labeling all states of $\smp$
with the letters of $\Sigma$ by a suitable~\mbox{$L : S \rightarrow \Sigma$}.
Then, every run $\sigma = s_0\,t_0\,s_1\,t_1\cdots$ of 
$\smp$ determines a unique
timed word $w_\sigma = L(s_0)\,t_0\,L(s_1)\,t_1\cdots$, and one can easily 
show that for every timed word~$w$, the set 
\mbox{$\{\sigma \in \setofruns \mid w_\sigma = w\}$} is measurable 
in $(\setofruns,\sigmafield,\probm)$.

\section{DTA Measures over SMPs}\label{sec:th-res}
Throughout this section we fix an SMP 
$\smp = (S,\prob, \D, \distribution_0)$ and a DTA 
$\calA = (Q, \Sigma, \calX, {\longrightarrow}, q_0)$ where
$\calX = \{x_1,\ldots,x_n\}$.
To simplify our notation, we assume that $\Sigma = S$, i.e., 
every run~$\sigma$ of~$\smp$ is a timed word over~$\Sigma$ (hence,
we do not need to introduce any labeling $L : S \rightarrow \Sigma$).
This technical assumption does not affect the generality of our results
(all of our arguments and proofs work exactly as they are,
we only need to rewrite them using less readable notation).
Our goal is to prove the following:

\begin{theorem}\label{thm:mainthm}
~
\begin{enumerate}
\item\label{item:mainthm-existence}
$\dme^{\calA}$ is well-defined for almost all runs of~$\smp$.
\item \label{item:mainthm-partition}
  There are pairwise disjoint sets  
  $\setofruns_1,\ldots,\setofruns_k$ of runs in $\smp$ such that 
  $\probm(\setofruns_1 \cup \cdots \cup \setofruns_k)=1$, and for
  every \mbox{$1 \leq j \leq k$} there is a tuple $D_j$ such that
  $\dme^{\calA}(\sigma) = D_j$ for almost all \mbox{$\sigma \in \setofruns_j$}
  (we use $D_{j,q}$ to denote the \mbox{$q$-component} of $D_j$). 
\end{enumerate}
\end{theorem}
In Section~\ref{SEC:ALGORITHMS}, we show how to compute the~$k$ and
approximate $\probm(\setofruns_j)$ and $D_j$ up to an
arbitrarily small given precision.

An immediate corollary of Theorem~\ref{thm:mainthm} is an analogous
result for $\cme^{\calA}$.
\begin{corollary}\label{COR:CONT-MES}
  $\cme^{\calA}$ is well-defined for almost all runs of~$\smp$. Further,
  there are pairwise disjoint sets  
  $\setofruns_1,\ldots,\setofruns_K$ of runs in $\smp$ such that 
  $\probm(\setofruns_1 \cup \cdots \cup \setofruns_K)=1$, and for
  every \mbox{$1 \leq j \leq K$} there is a tuple $C_j$ such that
  $\cme^\calA(\sigma) = C_j$ for almost all \mbox{$\sigma \in \setofruns_j$}.
\end{corollary}
Corollary~\ref{COR:CONT-MES} follows from Theorem~\ref{thm:mainthm} simply
by considering the discrete $\dme^{S \times \calA}$ measure, where the DTA
$S\times\calA$ is obtained from $\calA$ in the following way: the set of
locations of $S \times \calA$ is $\{q_0\} \cup (S \times Q)$, and for every
transition $(q_0,s,g,X,q')$ of $\calA$ we add a transition
$(q_0,s,g,X,(s,q'))$ to $S \times \calA$ and for every transition
$(q,s,g,X,q')$ and every $s' \in S$ we add a transition
$((s',q),s,g,X,(s,q'))$ to $S \times \calA$. The initial location 
of $S \times \calA$ is $q_0$. Intuitively, $S \times \calA$ is the 
same as $\calA$ but it explicitly ``remembers'' the letter 
which was used to enter the current location.
Let $k$ and $D_j$ be the constants of Theorem~\ref{thm:mainthm} constructed
for $\smp$ and $S \times \calA$. Observe that the expected time 
of performing a transition from a given $s \in S$, denoted by $E_s$, 
is given by 
$E_s = \sum_{s'\in S} \prob(s)(s')\cdot E_{s,s'}$, where $E_{s,s'}$ is the
expectation of a random variable with the density $\D(s,s')$. 
From this we easily obtain that
\begin{equation}\label{eq:cmeasure}
  C_{j,q} = \frac{\sum_{s\in S} E_s \cdot D_{j,(s,q)}}%
                {\sum_{p\in Q} \sum_{s\in S} E_s \cdot D_{j,(s,p)}}
\end{equation}
for all $q \in Q$ and $1 \leq j \leq k$. The details are given 
in Appendix~\ref{proof:cor:cont-mes}. Hence, we can also compute
the constant $K$ and approximate $\probm(\setofruns_j)$ and $C_j$
for every $1 \leq j \leq K$ using Equation~(\ref{eq:cmeasure}).

It remains to prove Theorem~\ref{thm:mainthm}. Let us start by sketching
the overall structure of our proof. First, we construct a synchronous product
$\product$ of $\smp$ and $\calA$, which is a Markov chain
with an uncountable state space \mbox{$\Gamma_\product = S \times Q \times
(\Rsetpo)^n$}.  Intuitively, $\product$ behaves in the same way as $\smp$
and simulates the computation of $\calA$ on-the-fly (see
Figure~\ref{fig:product}).
\begin{figure}[t]
\usetikzlibrary{snakes}
\begin{center}
\begin{tikzpicture}[->,>=stealth', node distance=0.5cm,
state/.style={shape=ellipse,draw, font=\footnotesize,inner sep=0.5mm,outer sep=0.8mm},
product/.style={shape=rectangle,draw, font=\footnotesize,inner sep=1mm,outer sep=0.8mm},
wait/.style={font=\scriptsize,above},
time/.style={font=\scriptsize,above left},
letter/.style={font=\scriptsize,left,decorate,decoration={snake,amplitude=0.3mm,segment length=1mm,post length=1.2mm}}
]
\matrix[row sep=0.5cm]{
 \node (SMP) {$\smp:$}; &[0.1cm]
 \node (s0) [state] {$s_0$}; &[0.7cm]
 \node (s1) [state] {$s_1$}; &[0.7cm]
 \node (s2) [state] {$s_2$}; \\
 \node (TA) {$\ta:$}; &
 \node (q0) [state] {$q_0,\nu_0$}; &
 \node (q1) [state] {$q_1,\nu_1$}; &
 \node (q2) [state] {$q_2,\nu_1$}; \\[1cm]
 \node (P) {$\product:$}; &
 \node (z0) [product] {$s_0,q_0,\nu_0$}; &
 \node (z1) [product] {$s_1,q_1,\nu_1$}; &
 \node (z2) [product] {$s_2,q_2,\nu_2$}; \\
};
\node (q1') at ([yshift=-1.2cm]q0) [state] {$q_1,\bar{\nu}_0$};
\node (q2') at ([yshift=-1.2cm]q1) [state] {$q_2,\bar{\nu}_1$};
\path (s0) edge [wait] node[wait] {$t_0$} (s1);
\path (s1) edge [wait] node[wait] {$t_1$} (s2);
\path (q0) edge [letter] node[letter] {$s_0$} (q1');
\path (q1') edge [time] node[time] {$t_0$} (q1);
\path (q1) edge [letter] node[letter] {$s_1$} (q2');
\path (q2') edge [time] node[time] {$t_1$} (q2);

\path (z0) edge [wait] node[wait] {$t_0$} (z1);
\path (z1) edge [wait] node[wait] {$t_1$} (z2);
\end{tikzpicture}
\end{center}

\caption{Synchronizing $\smp$ and $\calA$ in $\product$. 
Notice that $\nu_0 = \bar{\nu}_0 = \vec{0}$ and
$\nu_{i+1} = \bar{\nu_{i}} + t_{i}$.
}
\label{fig:product}
\end{figure}
Then, we construct a finite region graph
$\regiongraph_\product$ over the
product $\product$.  The nodes of $\regiongraph_\product$ are the sets 
of states that, roughly speaking, satisfy the same guards of $\calA$. Edges are
induced by transitions of the product (note that if two states satisfy the
same guards, the sets of enabled outgoing transitions are the same).
By relying on arguments presented 
in~\cite{ACD:verifying-automata-real-time,BKKKR:GSMP-games-TA}, we show 
that almost all runs reach a node of a bottom strongly connected 
component (BSCC) $\mathcal{C}$ of $\regiongraph_\product$ (by definition, 
each run which enters $\mathcal{C}$ remains in
$\mathcal{C}$). This gives us the partition of the set of runs of $\smp$
into the sets $\setofruns_1,\ldots,\setofruns_k$ (each $\setofruns_j$ 
corresponds to one of the BSCCs of $\regiongraph_\product$).

Subsequently, we concentrate on a fixed BSCC $\mathcal{C}$, and prove
that almost all runs that reach $\mathcal{C}$ have the 
same frequency of visits to a given $q\in Q$ (this gives us the 
constant $D_{j,q}$). Here we employ several deep results from the 
theory of general state space Markov chains 
(see Theorem~\ref{thm:gen_inv}). 
To apply these results, we prove that assuming aperiodicity 
of $\regiongraph_\product$ (see Definition~\ref{def:period}),
the state space of the product $\product$ is
\emph{small} (see Definition~\ref{def:small} and
Lemma~\ref{prop:small-product} below). This is perhaps the most 
demanding part of our proof. Roughly speaking, we show that
there is a distinguished subset of states reachable from each state 
in a fixed number of steps with probability bounded from~$0$. 
By applying Theorem~\ref{thm:gen_inv}, we
obtain a complete invariant distribution on the product, i.e., 
in principle, we obtain a constant frequency of any non-trivial subset 
of states. From this we derive our results in a straightforward way.
If $\regiongraph_\product$ is periodic, we use standard techniques 
for removing periodicity and then basically follow the same 
stream of arguments as in the aperiodic case.

\subsection{General state space Markov chains}
\label{sec:gssmc}

We start by recalling the definition of ``ordinary'' discrete-time 
Markov chains with discrete state space (DTMC). A DTMC is given
by a finite or countably infinite state space $S$, an initial probability
distribution over $S$, and a one-step transition matrix $P$ which defines
the probability $P(s,s')$ of every transion $(s,s') \in S \times S$ so that
$\sum_{s' \in S} P(s,s') = 1$ for every $s \in S$. In the setting of 
uncountable state spaces, transition probabilities cannot be specified
by a transition matrix. Instead, one defines the probabilities of moving
from a given state $s$ to a given measurable subset $X$ of states. 
Hence, the concept of transition matrix is replaced with a more
general notion of \emph{transition kernel} defined below. 

\begin{definition}\label{def:kernel}
  A \emph{transition kernel} over a measurable space $(\Gamma,\G)$ is a
  function $\kernel:\Gamma\times\G\to[0,1]$ such that
  \begin{enumerate}
  \item $\kernel(z,\cdot)$ is a probability measure over $(\Gamma,\G)$ for
    each $z\in\Gamma$;
  \item $\kernel(\cdot,A)$ is a measurable function for each $A\in\G$
    (i.e., for every $c\in \Rset$, the set of all $z\in \Gamma$ satisfying
    \mbox{$\kernel(z,A)\geq c$} belongs to $\G$).
   \end{enumerate}
\end{definition}
A transition kernel is the core of the following definition.
\begin{definition}\label{def:gssmc}
  A \emph{general state space Markov chain (GSSMC)} with a state space
  $(\Gamma,\G)$, a transition kernel $\kernel$ and an initial probability
  measure $\initmeasure$ is a stochastic process $\Phi=\Phi_1,\Phi_2,\ldots$
  such that each $\Phi_i$ is a random variable over a probability space
  $(\Omega_{\Phi},\sigmafield_{\Phi},\probm_{\Phi})$ where
  \begin{itemize}
  \item $\Omega_{\Phi}$ %
    is a set of \emph{runs}, i.e., infinite words over~$\Gamma$.
  \item $\sigmafield_{\Phi}$ is the product $\sigma$-field
    $\Productfield_{i=0}^{\infty} \G$.
  \item $\probm_{\Phi}$ is the unique probability measure over
    $(\Omega_{\Phi},\sigmafield_{\Phi})$ such that for every finite sequence 
    $A_0,\cdots,A_n \in \sigmafield_{\Phi}$ we have that
    $\probm_\Phi(\Phi_0 {\in} A_0, \cdots,\Phi_n {\in} A_n)$ is equal to
    \begin{equation}\label{eq:gssmc}
      \int_{y_0 \in A_0}\!\cdots \int_{y_{n-1} \in A_{n-1}}
      \initmeasure(\de{y_0})\cdot \kernel(y_0,\de{y_1}) \cdots
      \kernel(y_{n-1},A_n).
    \end{equation}
  \item Each $\Phi_i$ is the projection of elements of $\Omega_{\Phi}$ onto
    the $i$-th component.
  \end{itemize}
\end{definition}
A \emph{path} is a finite sequence $z_1 \cdots z_n$ of states from $\Gamma$.
From Equation~(\ref{eq:gssmc}) we get that $\gssmc$ also
satisfies the following properties which will be used to show several
results about the chain $\gssmc$ by working with the transition kernel only.
\begin{enumerate}
\item $\probm_{\Phi} (\Phi_0\in A_0)=\initmeasure(A_0)$,
\item $\probm_{\Phi} (\Phi_{n+1}\in A\mid \Phi_n,\ldots,\Phi_0) =
  \probm_{\Phi} (\Phi_{n+1}\in A\mid \Phi_n) = \kernel(\Phi_n,A)$ almost
  surely,
\item $\probm_{\Phi} (\Phi_{n+m}\in A\mid \Phi_n)= \kernel^m(\Phi_n,A)$
  almost surely,
\end{enumerate}
where the $m$-step transition kernel $\kernel^m$ is defined as follows:
\begin{align*}
  \kernel^1(z,A)    &= \kernel(z,A) \\
  \kernel^{i+1}(z,A) &= \int_{\Gamma} \kernel(z,dy) \cdot \kernel^i(y,A).
\end{align*}
Notice that the transition kernel and the $m$-step transition kernel 
are analogous
counterparts to the transition matrix and the $k$-step transition matrix 
of a DTMC.

As we mentioned above, our proof of Theorem~\ref{thm:mainthm} employs 
several results of GSSMC theory. In particular, we make use of the notion of
\emph{smallness} of the state space defined as follows.
\begin{definition}\label{def:small}
  Let $m\in \Nset$, $\varepsilon>0$, and $\nu$ be a probability measure
  on $\G$.  A set $C\in \G$ is
  $(m,\varepsilon,\nu)$-\emph{small} if for all $x\in C$ and $B\in \G$
  we have that $\kernel^m(x,B) \geq \varepsilon\cdot\nu(B)$.
\end{definition}
GSSMCs where the whole state space is small have many nice properties, 
and the relevant ones are summarized in the following theorem.
\begin{theorem}\label{thm:gen_inv}
  If $\Gamma$ is $(m,\varepsilon,\nu)$-small, then
  \begin{enumerate}
  \item \label{thm:gen_inv_1}{\bf [Existence of invariant measure]} There
    exists a unique probability measure $\pi$ such that for all $A\in
    \G$ we have that
    \[
    \pi(A)\quad = \quad \int_{\Gamma} \pi(dx)\kernel(x,A)
    \]
  \item \label{thm:gen_inv_2}{\bf [Strong law of large numbers]} If
    $h:\Gamma\rightarrow \Rset$ satisfies $\int_{\Gamma} h(x)
    \pi(dx)<\infty$, then almost surely
    \[
    \lim_{n\rightarrow \infty} \frac{\sum_{i=1}^n h(\Phi_i)}{n}\quad = \quad
    \int_{\Gamma} h(x) \pi(dx)
    \]
  \item \label{thm:gen_inv_3}{\bf [Uniform ergodicity]} For all $x\in
    \Gamma$, $A \in \G$, and all $n\in \Nset$,
    \[
    \sup_{A\in \G} |\kernel^n(x,A)-\pi(A)|\quad \leq \quad
    (1-\varepsilon)^{\lfloor n/m\rfloor}
    \]
  \end{enumerate}
\end{theorem}
\begin{proof}
  The theorem is a consequence of stadard results for GSSMCs.
Since $\Gamma$ is $(m,\varepsilon,\nu)$-small, we have
  \begin{enumerate}[(i)]
  \item $\Phi$ is by definition $\varphi$-irreducible for $\varphi=\nu$, and
    thus also $\psi$-irreducible by~\cite[Proposition~4.2.2]{MT:book};
  \item $\Gamma$ is by definition also $(a,\varepsilon,\nu)$-petite (see~\cite[Section 5.5.2]{MT:book}),
      where $a$ is the Dirac distribution on $\Nseto$ with $a(m)=1$,
      $a(n)=0$ for $n\neq m$;
  \item the first return time to $\Gamma$ is trivially $1$.
  \end{enumerate}
  \begin{enumerate}
  \item[ad 1.] By (iii), $\Gamma$ is not uniformly transient, hence by (i),
    (ii) and \cite[Theorem 8.0.2]{MT:book}, $\Phi$ is recurrent. Thus
    by~\cite[Theorem~10.0.1]{MT:book}, there exists a unique invariant
    probability measure $\pi$.
  \item[ad 2.] By (i)-(iii) and~\cite[Theorem~10.4.10 (ii)]{MT:book},
    $\Phi$ is positive Harris. Therefore, we may apply~\cite[Theorem~17.0.1
    (i)]{MT:book} and obtain the desired result.
  \item[ad 3.] This follows immediately
    from~\cite[Theorem~8]{RR:GSSMC-PS}.\qed
  \end{enumerate}
\renewcommand{\qed}{}
\end{proof}

\subsection{The product process}\label{SEC:PRODUCT-PROCESS}

The \emph{product process} of $\smp$ and $\ta$, denoted by $\product$,
is a GSSMC with the state space 
$\Gamma_\product = S \times Q \times (\Rsetpo)^n$, where
$n = |\calX|$ is the number of clocks of~$\calA$.  The
$\sigma$-field over $\Gamma_\product$ is the product $\sigma$-field
\mbox{$\G_\product = 2^S \productfield 2^Q \productfield \borel^n$}
where $\borel^n$ is the
Borel $\sigma$-field over the set $(\Rsetpo)^n$.  For each $A\in
\G_\product$, the initial probability $\initmeasure_\product(A)$ is equal to
$\sum_{(s,q_0,\vec{0})\in A} \distribution_0(s)$ (recall that $\distribution_0$ is the initial distribution
of $\smp$).

The behavior of $\product$ is depicted in Figure~\ref{fig:product}.
Each step of the product process corresponds to one step of
$\smp$ and two steps of $\ta$. The step of the product starts by 
simulating the discrete step of $\ta$ that reads the 
current state of $\smp$ and possibly resets some clocks, followed
by simulating simultaneously the step of $\smp$ that takes time $t$ and 
the corresponding step of $\ta$ which reads the time stamp~$t$. 

Now we define the transition kernel $\kernel_\product$ of the product
process. Let $z = (s,q,\nu)$ be a state of $\Gamma_\product$,
and let $(\bar{q},\bar{\nu})$ be the configuration of $\calA$ entered
from the configuration $(q,\nu)$ after reading~$s$ (note that
$\bar{\nu}$ is not necessarily the same as $\nu$ because $\calA$ 
may reset some clocks). 
It suffices to define $\kernel_\product(z,\cdot)$ only for generators of $\G_\product$ and
then apply the extension theorem (see, e.g., \cite{Billingsley:book}) to obtain a unique probability measure 
$\kernel_\product(z,\cdot)$ over $(\Gamma_\product,\G_\product)$. Generators of $\G_\product$ are
sets of the form $\{s'\}\times \{q'\}\times \vec{I}$ where $s'\in S$, $q'\in Q$ and $\vec{I}$ is the product
$I_1\times \cdots \times I_n$ of intervals $I_i$ in $\Rsetpo$. If 
$q' \neq \bar{q}$, then we define $\kernel_\product(z,\{s'\}\times \{q'\}\times \vec{I}) = 0$. Otherwise,
we define
\[
  \kernel_\product(z,\{s'\}\times \{q'\}\times \vec{I}) = \prob(s)(s')\cdot \int_0^\infty f(t) \cdot 1_\vec{I}(\bar{\nu} + t) 
  \de{t}
\]
Here $f = \D(s,s')$ and $1_\vec{I}$ is the indicator function of 
the set~$\vec{I}$.

Since $\kernel_\product(z,\cdot)$ is by definition a probability measure
over $(\Gamma_\product,\G_\product)$, it remains to 
check the second condition of Definition~\ref{def:kernel}.

\begin{lemma}\label{lem:product-kernel}
  Let $A \in \G_\product$. Then $\kernel_\product(\cdot,A)$ is a
  measurable function, i.e., $\product$ is a GSSMC.
\end{lemma}
A proof of this lemma can be found in Appendix~\ref{app:product-kernel}.
Recall that by Definition~\ref{def:gssmc}, $\probm_\product$ is the unique
probability measure on the product $\sigma$-field
$\sigmafield_\product=\Productfield_{i=0}^{\infty} \G_\product$ induced by 
$\kernel_\product$ and the initial probability measure
$\initmeasure_\product$.

\subsubsection{The correspondence between $\product$ and $\smp$}

In this subsection we show that $\product$ correctly reflects the
behaviour of $\smp$. First, we define the $\dme^{\calA}$ measure for
$\product$. (As the DTA $\calA$ is fixed, we omit them and write
$\dme$ and $\dme_q$ instead of $\dme^{\calA}$ and $\dme^{\calA}_q$,
respectively.) Let $\sigma = (s_0,q_0,\nu_0)\,(s_1,q_1,\nu_1)\cdots$
be a run of $\product$ and $q\in Q$ a location. For every $i \in \Nseto$,
let $1^i_{q}(\sigma)$ be either $1$ or $0$ depending on whether
if $q_i=q$ or not, respectively. We put
\[
 \dme_{q}(\sigma) = \limsup_{n\to\infty} \frac{\sum_{i=1}^n 1^i_{q}(\sigma)}{n}
\]  
\begin{lemma}\label{lem:product-correctness}
  There is a measurable one-to-one mapping $\pr$ from the set of runs
  of $\smp$ to the set of runs of $\product$ such that
  \begin{itemize}
  \item $\pr$ preserves measure, i.e., for every measurable set $X$ 
    of runs of $\smp$ we have that $\pr(X)$ is also measurable and 
    $\probm_\smp(X) = \probm_\product(\pr(X))$;
  \item $\pr$ preserves $\dme$, i.e., for every run $\sigma$ of $\smp$
    and every $q \in Q$ we have that $\dme_{q}(\sigma)$ is well-defined 
    iff $\dme_{q}(\pr(\sigma))$ is well-defined, 
    and $\dme_{q}(\sigma) = \dme_{q}(\pr(\sigma))$.
  \end{itemize}
\end{lemma}
A formal proof of Lemma~\ref{lem:product-correctness} is given
in Appendix~\ref{app:product-correctness}.

\subsubsection{The region graph of $\product$}\label{sec:region-graph}
Although the state-space $\Gamma_\product$ is uncountable, we can define the
standard \emph{region relation} $\region$ \cite{AD:Timed-Automata} over
$\Gamma_\product$ with finite index, and then work with finitely many
\emph{regions}.
For a given $a \in \Rset$, we use $\fr(a)$ to denote the fractional part
of~$a$, and $\mathit{int}(a)$ to denote the integral part of~$a$.
For $a,b \in \Rset$, we say that $a$ and $b$ \emph{agree on integral part}
if $\mathit{int}(a) = \mathit{int}(b)$ and neither 
or both $a$, $b$ are integers.

We denote by $\maxb$ the maximal constant that appears in the guards of $\ta$
and say that a
clock $x \in \calX$ is \emph{relevant} for $\nu$ if $\nu(x) \leq \maxb$.
Finally, we put $(s_1,q_1,\nu_1) \region (s_2,q_2,\nu_2)$ if
\begin{itemize}
\item $s_1 = s_2$ and $q_1 = q_2$;
\item for all relevant $x \in \calX$ we have that $\nu_1(x)$ and $\nu_2(x)$
  agree on integral parts;
\item for all relevant $x, y \in \calX$ we have that \mbox{$\fr(\nu_1(x)) \leq
  \fr(\nu_1(y))$} iff $\fr(\nu_2(x)) \leq \fr(\nu_2(y))$.
\end{itemize}
Note that $\region$ is an equivalence with finite index. The equivalence
classes of $\region$ are called \emph{regions}.
Observe that states in the same region have the
same behavior with respect to qualitative reachability. This is formalized
in the following lemma.
\begin{lemma}
  \label{lem:action-in-region-leads-to-same-non-zero-regions}
  Let $R$ and $T$ be regions and $z,z'\in R$. Then 
  $\kernel_\product(z,T) > 0$ iff $\kernel_\product(z',T) > 0$.
\end{lemma}
A proof of Lemma~\ref{lem:action-in-region-leads-to-same-non-zero-regions}
can be found in \cite{BKKKR:GSMP-games-TA}. Further, 
we define a finite \emph{region graph} 
$\regiongraph_\product = (V,E)$ where the set of
vertices $V$ is the set of regions and for every pair of regions
$R,R'$ there is an edge $(R,R') \in E$ iff \mbox{$\kernel_\product(z,R') > 0$}
for some $z \in R$ (due to 
Lemma~\ref{lem:action-in-region-leads-to-same-non-zero-regions}, the
concrete choice of $z$ is irrelevant). For technical reasons, we assume 
that $V$ contains only regions reachable with positive probability in 
$\product$.

\subsection{Finishing the proof of 
Theorem~\ref{thm:mainthm}}\label{SEC:FINISH-MAIN-THM}

Our proof is divided into three parts. In the first part we consider
a general region graph which is not necessarily strongly connected,
and show that we can actually concentrate just on its BSCCs. In the
second part we study a given BSCC under the aperiodicity assumption.
Finally, in the last part we consider a general BSCC which may 
be periodic. (The second part is included mainly for the sake of 
readability.)

\subsubsection*{Non-strongly connected region graph}

Let $\mathcal{C}_1,\ldots,\mathcal{C}_k$ be the BSCCs of the region graph.
The set $\setofruns_i$ consists of all runs $\sigma$ of~$\smp$ such that
$\pr(\omega)$ visits (a configuration in a region of) $\mathcal{C}_i$, where 
$\pr$ is the mapping of Lemma~\ref{lem:product-correctness}.
By applying the arguments of 
\cite{ACD:verifying-automata-real-time,BKKKR:GSMP-games-TA},
it follows that almost runs in $\product$ visit a configuration 
of a BSCC. By Lemma~\ref{lem:product-correctness}, $\pr$ preserves
$\dme$ and the probability $\probm_{\smp}(\setofruns_i)$ is equal 
to the probability of visiting $\mathcal{C}_i$ in $\product$. 
Further, since the value of $\dme$ does not depend on a finite
prefix of a run, we may safely assume that $\product$ is initialized 
in $\mathcal{C}_i$ in such a way that the initial distribution 
corresponds to the conditional distribution of the first visit to 
$\mathcal{C}_i$ conditioned on visiting $\mathcal{C}_i$.

In a BSCC $\mathcal{C}_i$, there may be some \emph{growing} clocks 
that are never reset. Since the values of growing clocks are just 
constantly increasing, the product process never returns to a state 
it has visited before. Therefore, there is no invariant distribution. 
Observe that all runs initiated in $\mathcal{C}_i$ eventually reach 
a configuration where the values of all growing clocks are larger
than the maximal constant $\maxb$ employed in the guards of~$\calA$.
This means that $\mathcal{C}_i$ actually consists \emph{only} of regions
where all growing clocks are irrelevant
(see Section~\ref{sec:region-graph}), because  $\mathcal{C}_i$ would 
not be strongly connected otherwise. Hence, we can safely \emph{remove} 
every growing clock~$x$ from $\mathcal{C}_i$, replacing all guards 
of the form  $x > c$ or $x \geq c$ with $\mathit{true}$ and all 
guards of the form $x < c$ or $x \leq c$ with $\mathit{false}$.
So, from now on we assume 
that there are no growing clocks in~$\mathcal{C}_i$.

\subsubsection*{Strongly connected \& aperiodic region graph}

In this part we consider a given BSCC $\mathcal{C}_i$ of the region
graph $\regiongraph_\product$. This is equivalent to assuming that
$\regiongraph_\product$ is strongly connected and $\Gamma_\product$ is
equal to the union of all regions of $\regiongraph_\product$ (recall
that $\regiongraph_\product$ consists just of regions reachable with
positive probability in $\product$). We also assume that there are no
growing clocks (see the previous part).  Further, in this subsection
we assume that $\regiongraph_\product$ is aperiodic in the following
sense.
\begin{definition}\label{def:period}
  A \emph{period} $p$ of the region graph $\regiongraph_\product$ is the
  greatest common divisor of lengths of all cycles in
  $\regiongraph_\product$. The region graph $\regiongraph_\product$ is
  \emph{aperiodic} if $p=1$.
\end{definition}

The key to proving Theorem~\ref{thm:mainthm} in the current restricted
setting is to show that the state space of $\product$ is small
(recall Definition~\ref{def:small}) and then apply
Theorem~\ref{thm:gen_inv} (\ref{thm:gen_inv_1}) 
and~(\ref{thm:gen_inv_2}) to obtain the required characterization of the
long-run behavior of~$\product$.
\begin{proposition}\label{prop:small-product}
  Assume that $\regiongraph_{\product}$ is strongly connected and
  aperiodic.  Then there exist a region $R$, a measurable 
  subset $S\subseteq R$, $n\in\Nset$, $b>0$, and a probability 
  measure $\smallmeasure$ such that
  $\smallmeasure(S)=1$ and for all measurable $T\subseteq S$ and
  $z\in\Gamma_\product$ we have that
  $\kernel_\product^{n}(z,T)>b\cdot\smallmeasure(T)$.  In other words, the set
  $\Gamma_\product$ of all states of the GSSMC $\product$ is
  $(n,b,\smallmeasure)$-small.
\end{proposition}

\begin{proof}[Sketch]
  We show that there exist \mbox{$z^\ast \in\Gamma_\product$}, 
  $n \in \Nset$, and $\gamma > 0$ such that for an arbitrary 
  starting state $z\in\Gamma_\product$ there is a path from $z$
  to $z^\ast$ of length exactly~$n$ that is \emph{$\gamma$-wide}
  in the sense that the waiting time of any transition
  in the path can be changed by $\pm\gamma$ without ending up in a
  different region in the end.  The target set $S$ then corresponds
  to a ``neighbourhood'' of $z^\ast$ within the region of~$z^\ast$.
  Any small enough sub-neighbourhood of $z^\ast$ is visited by 
  a set of runs that follow the $\gamma$-wide path closely enough.  
  The probability of this set of runs then depends linearly on the size of the
  sub-neighbourhood when measured 
  by~$\smallmeasure$, where $\smallmeasure$ is essentially
  the Lebesgue measure restricted to~$S$.

  So, it remains to find suitable $z^\ast$, $n$, and $\gamma$.
  For a given starting state \mbox{$z \in\Gamma_\product$}, we 
  construct a path of fixed
  length~$n$ (independent of $z$) that always 
  ends in the same state $z^\ast$. Further, the path is $\gamma$-wide
  for some  $\gamma > 0$ independent of~$z$. Technically, the path
  is obtained by concatenating five sub-paths each of which has a fixed
  length independent of~$z$. These sub-paths are described in greater
  detail below.

  In the first sub-path, we move to a \emph{$\delta$-separated} state for some
  fixed $\delta > 0$ independent of~$z$. A state is $\delta$-separated 
  if the fractional parts of all relevant clocks are approximately 
  equally distributed
  on the $[0,1]$ line segment (each two of them have distance at least
  $\delta$).  We can easily build the first sub-path so that it is 
  $\delta$-wide.

  For the second sub-path, we first fix some region $R_1$. Since
  $\regiongraph_\product$ is strongly connected and aperiodic,
  there is a fixed~$n'$ such that $R_1$ is reachable from 
  an arbitrary state of $\Gamma_\product$ in exactly $n'$~transitions.
  The second sub-path is chosen as a $(\delta/n')$-wide path of 
  length $n'$ that leads to a $(\delta/n')$-separated state
  of $R_1$ (we show that such a sub-path is guaranteed to exist;
  intuitively, the reason why the separation and wideness may 
  decrease proportionally to~$n'$ is that the fractional parts of 
  relevant clock may be forced to move closer and closer to each other 
  by the resets performed along the sub-path).

  In the third sub-path, we squeeze the fractional parts of all relevant
  clocks close to~$0$.  We go through a fixed region path $R_1 \cdots
  R_k$ (independent of $z$) so that in each step we shift the 
  time by an integral value minus a small constant~$c$ (note that
  the fractional parts of clocks reset during this path
  have fixed relative distances).
  Thus, we reach a state $z'_k$ that is ``almost fixed'' in the sense
  that the values of all relevant clocks in $z'_k$ are the same
  for every starting state~$z$. Note that the third sub-path
  is $c$-wide. At this point, we should note that if we defined 
  the product process somewhat differently by identifying all states
  differing only in the values of irrelevant clocks (which does not
  lead to any technical complications), we would be done, i.e., 
  we could put $z^\ast = z'_k$. We have neglected this possibility 
  mainly for presentation reasons. So, we need two more sub-paths
  to fix the values of irrelevant clocks.

  In the fourth sub-path, we act similarly as in the first sub-path
  and prepare ourselves for the final sub-path.
  We reach a $\delta$-separated state that is
  almost equal to a fixed state $z_{\ell} \in R_{\ell}$.  
  Again, we do it by a $\delta$-wide path of a fixed length.

  In the fifth sub-path, we follow a fixed region path 
  $R_{\ell} \cdots R_{\ell+m}$ such that each clock not relevant 
  in $R_\ell$ is reset along this path, and hence we reach a fixed 
  state $z^\ast \in R_{\ell+m}$. Here we use our assumption that 
  every clock can be reset to zero (i.e., there are   
  no growing clocks).
\end{proof}

Now we may finish the proof of Theorem~\ref{thm:mainthm}.
By Theorem~\ref{thm:gen_inv}~(\ref{thm:gen_inv_1}), there 
is a unique invariant distribution $\pi$ on $\Gamma_\product$.
For every $q\in Q$, we denote by $A_q$ the set of all
states of $\product$ of the form $(s,q,\nu) \in \Gamma_\product$.  By
Theorem~\ref{thm:gen_inv}~(\ref{thm:gen_inv_2}), for almost all runs
$\sigma$ of $\product$ we have that $\dme(\sigma)$ is
well-defined and $\dme_{q}(\sigma)= \sum \pi(A_{q})$. By
Lemma~\ref{lem:product-correctness}, we obtain the same for almost
all runs of~$\smp$.

\subsubsection*{Strongly connected \& periodic region graph}

Now we consider a general BSCC $\mathcal{C}_i$ of the region
graph $\regiongraph_\product$. Technically, we adopt the same setup
as the previous part but remove the aperiodicity condition.
That is, we assume that $\regiongraph_\product$ is strongly 
connected, $\Gamma_\product$ is equal to the union of all regions 
of $\regiongraph_\product$, and there are no
growing clocks. 

Let $p$ be the period of $\regiongraph_\product$.  
In this case, $\product$ is not necessarily
small in the sense of Definition~\ref{def:small}. By
employing standard methods for periodic Markov chains,
we decompose $\product$ into $p$
stochastic processes $\Phi_0,\ldots,\Phi_{p-1}$ where each $\Phi_k$
makes steps corresponding to $p$ steps of the original process
$\product$ (except for the first step which corresponds just to~$k$ 
steps of $\product$).  Each $\Phi_k$ is aperiodic and hence small
(this follows by slightly generalizing the arguments of the 
previous part; see Proposition~\ref{prop:main-lem-period}). 
Thus, we can apply Theorem~\ref{thm:gen_inv} to each $\Phi_k$ 
separately and express the frequency of visits to $q$ in $\Phi_k$ in terms
of a unique invariant distribution $\pi_k$ for $\Phi_k$. Finally, we
obtain the frequency of visits to $q$ in
$\product$ as an average of the corresponding frequencies in $\Phi_k$.

Let us start by decomposing the set of nodes~$V$ of~$\regiongraph_\product$ 
into $p$ classes that constitute a cyclic structure (see e.g.~\cite[Theorem~4.1]{Bremaud:book}).
\begin{lemma}\label{lem:periodic-graph}
  There are disjoint sets $V_0,\ldots,V_{p-1}\subseteq V$ such that
  $V=\bigcup_{k=0}^{p-1} V_k$ and for all $u,v\in V$ we have that
  $(u,v)\in E$ iff there is $k\in \{0,\ldots,p-1\}$ satisfying $u\in
  V_k$ and $v\in V_{j}$ where $j=(k+1)\mod p$.
\end{lemma}
For each $k\in \{0,\ldots,p-1\}$ we construct a GSSMC $\Phi_k$ with
state space $\Gamma^k_\product = \bigcup_{R \in V_k} R$, a transition
kernel $\kernel^p(\cdot,\cdot)$ restricted to $\Gamma^k_\product$, and
an initial probability measure $\initmeasure_k$ defined by
$\initmeasure_k(A)=\int_{z\in \Gamma_\product} \initmeasure(dz)\cdot
\kernel^k(z,A)$.  For each $k$, we define the discrete frequency
$\dme_{q}^k$ of visits $q$ in the process
$\Phi_k$. Then we show that if $\dme^k$ is well-defined
in $\Phi_k$, we can express the frequency $\dme_{q}$ in~$\product$.

Note that for every run $z_0 \, z_1 \cdots$ of $\product$, the 
word $z_k\, z_{p+k}\, z_{2p+k}$ is a run of $\Phi_k$.
For a run $\sigma = (s_0,q_0,\nu_0)\, (s_1,q_1,\nu_1) \cdots$, 
$k \in \{0,\ldots,p-1\}$, and a location $q\in Q$, let
define $1^{i,k}_{q}(\sigma)$ to be either $1$ or $0$ depending on
whether $q_{ip+k}=q$ or not, respectively. Further, we put
\[
 \dme^k_{q}(\sigma) = \limsup_{n\to\infty} \frac{\sum_{i=1}^n 1^{i,k}_{q}(\sigma)}{n}
\]
Assuming that each $\dme^k$ is well-defined, for almost 
all runs $\sigma$ of $\product$ we have the following:
\begin{eqnarray*}
 \dme_{q}(\sigma) & = & 
\lim_{n\to\infty} \frac{\sum_{i=1}^n 1^{i}_{q}(\sigma)}{n} 
  \ = \  
\lim_{n\to\infty} \frac{\sum_{i=1}^n \sum_{k=0}^{p-1} 1^{i,k}_{q}(\sigma)}{np} \\
 & = & 
\frac{1}{p} \sum_{k=0}^{p-1} \lim_{n\to\infty} 
      \frac{\sum_{i=1}^n 1^{i,k}_{q}(\sigma)}{n} 
  \ = \  \frac{1}{p} \sum_{k=0}^{p-1} \dme_{q}^k(\sigma) 
\end{eqnarray*}
So, it suffices to concentrate on $\dme_{q}^k$. The following 
proposition is a generalization of Proposition~\ref{prop:small-product} 
to periodic processes.
\begin{proposition}\label{prop:main-lem-period}
  Assume that $\regiongraph_{\product}$ is strongly connected and has a
  period~$p$.  For every $k\in \{0,\ldots,p-1\}$ there exist a 
  region $R_k \in V_k$, a measurable $S_k \subset R_k$, $n_k\in\Nset$, 
  $b_k>0$, and a probability measure
  $\smallmeasure_k$ such that $\smallmeasure_k(S_k)=1$ and for every
  measurable $T\subseteq S_k$ and $z \in\Gamma^k_\product$ we have
  $\kernel_\product^{n_k\cdot p}(z,T)>b_k\cdot\smallmeasure_k(T)$.  In
  other words, $\Phi_k$ is \mbox{$(n_k,b_k,\smallmeasure_k)$-small}.
\end{proposition}
By Theorem~\ref{thm:gen_inv}~(\ref{thm:gen_inv_1}), for every $k\in
\{0,\ldots,p-1\}$, there is a unique invariant distribution $\pi_k$ on
$\Gamma_\product$ for the process $\Phi_k$.  By
Theorem~\ref{thm:gen_inv}~(\ref{thm:gen_inv_2}), each $\dme^k$ is
well-defined and for almost all runs $\sigma$ we have that
$\dme^k_{q}(\sigma)=\pi_k(A_{q})$. Thus, we obtain
\begin{equation*}
 \dme_{q}(\sigma)  =  \frac{1}{p} \sum_{k=0}^{p-1} \pi_k(A_{q})
\end{equation*}

\section{Approximating DTA Measures}\label{SEC:ALGORITHMS}
In this section we show how to approximate the DTA measures for SMPs
using the $m$-step transition kernel $\kernel^m_\product$ of $\product$.
The procedure for computing $\kernel^m_\product$ up to a sufficient precision
is taken as a ``black box'' part of the algorithm, we concentrate
just on developing generic bounds on~$m$ that are sufficient to achieve
the required precision.  

For simplicity, we assume that the initial distribution 
$\distribution_0$ of $\smp$ assigns $1$ to some $s_0\in S$ (all 
of the results presented in this section can easily be generalized 
to an arbitrary initial distribution). The initial state in $\product$
is $z_0 = (s_0,q_0,\vec{0})$.

As we already noted in the previous section, the constant~$k$
of~Theorem~\ref{thm:mainthm} is the number of BSCCs of 
$\regiongraph_\product$. For the rest of this section, we fix
some $1 \leq j \leq k$, and write just $\mathcal{C}$, $\setofruns$ and $D$
instead of $\mathcal{C}_j$, $\setofruns_j$ and $D_j$, respectively.
We slightly abuse our notation by using $\mathcal{C}$ to denote 
also the set of configurations that belong to some region of $\mathcal{C}$
(particularly in expressions such as $\kernel_\product(z,\mathcal{C})$).

The probability $\probm_\smp(\setofruns)$ is equal to the probability of
visiting~$\mathcal{C}$ in $\product$. Observe that 
\[
\probm_\smp(\setofruns) = 
\lim_{i \rightarrow \infty} \kernel^i_\product(z_0,\mathcal{C})
\]
Let us analyze the speed of this approximation. First, we
need to introduce several parameters.  Let $\pmin$ be
the smallest transition probability in~$\smp$, and $\delay(\smp)$ the
set of delay densities used in $\smp$,
i.e., $\delay(\smp)=\{\D(s,s')\mid s,s'\in S\}$. Let
$|V|$ be the number of vertices (regions) of $\regiongraph_\product$.
Due to our assumptions imposed on delay densities,
there is a fixed bound $\densb>0$ such that, for all
$f\in\delay(\smp)$ and $x\in[0,\maxb]$, either $f(x)>\densb$
or $f(x)=0$. Further, $\int_{\maxb}^\infty f(x)dx$ is either
larger than $\densb$ or equal to~$0$.

\newcommand{\tmreachability}{
For every $i\in \Nset$ we have that
\[
 \probm_\smp(\setofruns)-\kernel^i_\product(z_0,\mathcal{C})
\quad \leq \quad
\left( 1- \left(\frac{p_{min} \cdot \densb}%
{c}\right)^c
\right)^{\lfloor i/c\rfloor}%
\]
where $c = 4\cdot |V|$.
}
\begin{theorem}
\label{thm:reachability}
\tmreachability
\end{theorem}

\begin{proof}[Sketch]
We denote by $B$ the union of all regions that belong to BSCCs of $\regiongraph_\product$.  
We show that for $c = 4 \cdot |V|$ there is a lower bound 
$p_{\mathit{bound}} = (\pmin \cdot \densb \cdot 1/c)^{c}$ on the probability
of reaching $B$ in at most $c$ steps from any state $z \in
\Gamma_\product$. Note that then the probability of not hitting $B$
after $i= m\cdot c$ steps is at most $(1-p_{bound})^m$. However, this
means that $\kernel^i_\product(z,\calC)$ cannot differ from the
probability of reaching $\calC$ (and thus also from
$\probm_\smp(\setofruns)$) by more than $(1-p_{bound})^m$ because
$\calC\subseteq B$ and the probability of reaching $\calC$ from
$B\smallsetminus \calC$ is $0$.

The bound $p_{\mathit{bound}}$ is provided by arguments similar to the proof of 
Proposition~\ref{prop:small-product}. From any state $z$ we build a $\delta$-wide 
path to a state in $B$ that has length bounded by $4 \cdot |V|$ such that $\delta = \pmin \cdot \densb \cdot 1/c$. 
The paths that follow this $\delta$-wide path closely enough (hence, reach $B$) 
have probability $p_{\mathit{bound}}$.
\end{proof}

Now let us concentrate on approximating the tuple $D$. This can be
done by considering just the BSCC~$\mathcal{C}$. Similarly as in
Section~\ref{sec:th-res}, from now on we assume that $\calC$ is the
set of nodes of $\regiongraph_\product$ (i.e., $\regiongraph_\product$
is strongly-connected) and that $\Gamma_\product$ is equal to the
union of all regions of $\calC$.

As in Section~\ref{sec:th-res}, we start with the aperiodic case. 
Then, Theorem~\ref{thm:gen_inv}~(\ref{thm:gen_inv_3}.) implies
that each $D_q$ can be approximated using
$\kernel^i_\product(u,A_q)$ where $u$ is an arbitrary state of
$\Gamma_\product$ and $A_q$ is the set of all states of $\product$
of the form $(s,q,\nu)$. More precisely, we obtain the following:
\begin{theorem}
Assume that $\regiongraph_\product$ is strongly connected and aperiodic.
Then for all $i\in \Nset$, $u \in \Gamma_\product$, and $q\in Q$
\begin{align*}
\left|D_q- \kernel^i_\product(u,A_q)\right| \  &\leq \
\left(1-\left(\frac{p_{\min} \cdot \densb}{%
r}\right)^{r}\right)^{\lfloor i/r \rfloor}
\end{align*}
where $r=\lfloor |V|^{4 \ln |V|} \rfloor$.
\end{theorem}
\begin{proof}
From the proof of Proposition~\ref{prop:main-lem-period} (for
details see Appendix~\ref{app:sec-main-thm}), we obtain that 
$\Gamma_{\product}$ is $(m,\varepsilon,\smallmeasure)$-small 
with $m \leq r$ and
$\varepsilon = (\frac{\pmin \densb}{r}) ^{r}$, and the result follows
from Theorem~\ref{thm:gen_inv}~(\ref{thm:gen_inv_3}.).
\end{proof}

Now let us consider the general (periodic) case. We adopt the same
notation as in Section~\ref{sec:th-res}, i.e., the period of
$\regiongraph_\product$ is denoted by $p$, the decomposition
of the set $V$ by $V_0,\ldots,V_{p-1}$ (see Lemma~\ref{lem:periodic-graph}),
and $\Gamma^k_\product$ denotes the set
$\bigcup_{R\in V_k} R$ for every $k\in \{0,\ldots,p-1\}$.

\begin{theorem}
For every $i\in \Nset$ we have that
\begin{align*}
\left|D_q-\frac{1}{p} \cdot \sum_{k=0}^{p-1}
\kernel^{i \cdot p}_\product(u_k,A_q)\right| \  &\leq \
\left(1-\left(\frac{p_{\min} \cdot \densb}{%
r}\right)^{r}\right)^{\lfloor i/r \rfloor}
\end{align*}
where $u_k \in \Gamma^k_\product$ and  $r=\lfloor |V|^{4 \ln |V|} \rfloor$.
\end{theorem}
\begin{proof}
Due to the results of Section~\ref{sec:th-res} we have that
$D_q = \frac{1}{p} \cdot \sum_{k=0}^{p-1} \pi_k(A_q)$, where $\pi_k$
is the invariant measure for the $k$-th aperiodic 
decomposition $\Phi_k$ of the product process $\product$ (i.e.
$\pi_k$ is a measure over $\Gamma^k_\product$).
From the proof of Proposition~\ref{prop:main-lem-period} (for
details see Appendix~\ref{app:sec-main-thm}), $\Gamma^k_\product$
is $(m,\varepsilon,\smallmeasure)$-small with $m\leq r$ and
$\varepsilon = (\frac{\pmin \densb}{r})^{r}$, and the result follows
from Theorem~\ref{thm:gen_inv}~(\ref{thm:gen_inv_3}.) applied to
each $\Gamma^k_\product$ separately.
\end{proof}

\section{Conclusions}

We have shown that DTA measures over semi-Markov processes are
well-defined for almost all runs and assume only finitely 
many values with positive probability. We also indicated how to 
approximate DTA measures and the associated probabilities up to
an arbitrarily small given precision. 

Our approximation algorithm is quite naive and there is a lot of
space for further improvement. An interesting open question is whether
one can design more efficient algorithms with low complexity in the
size of SMP (the size of DTA specifications should stay relatively
small in most applications, and hence the (inevitable) exponential 
blowup in the size of DTA is actually not so problematic).

Another interesting question is whether the results presented in this paper
can be extended to more general stochastic models such as generalized
semi-Markov processes.

\subsubsection*{Acknowledgement}
The authors thank Petr Slov\'{a}k for his many
useful comments. 
The work has been supported by the 
Institute for Theoretical Computer Science, project No.~1M0545, 
and the Czech Science Foundation, grant No.~P202/10/1469 (T.~Br\'{a}zdil,
A.~Ku\v{c}era), No.~201/08/P459 (V.~\v{R}eh\'ak), and No.~102/09/H042 (J. Kr\v{c}\'al).

\bibliographystyle{plain}
\bibliography{str-long,concur}

\newpage
\appendix

\section{Proof of Corollary~\ref{COR:CONT-MES}}\label{proof:cor:cont-mes}

\begin{refcorollary}{COR:CONT-MES}
  $\cme^{\calA}$ is well-defined for almost all runs of~$\smp$. Further,
  there are pairwise disjoint sets $\setofruns_1,\ldots,\setofruns_K$ of
  runs in $\smp$ such that $\probm(\setofruns_1 \cup \cdots \cup
  \setofruns_K)=1$, and for every \mbox{$1 \leq j \leq K$} there is a tuple
  $C_j$ such that $\cme^\calA(\sigma) = C_j$ for almost all \mbox{$\sigma \in
    \setofruns_j$}.
\end{refcorollary}
Most of the proof has already been presented in Section~\ref{sec:th-res}.
It remains to prove that for almost all runs $\sigma$ of $\setofruns_j$ we have
\begin{equation}\label{eq:finale}
  \frac{\sum_{s\in S} E_s \cdot D_{j,(s,q)}}{\sum_{p\in Q} \sum_{s\in S} E_s\cdot D_{j,(s,p)}} = \cme^\calA_q(\sigma)
\end{equation}
To simplify our notation we write $D_{s,p}$ and $1^i_{s,p}$ instead of
$D_{j,(s,p)}$ and $1^i_{(s,p)}$, respectively.  If $D_q=0$ then clearly both
sides of Equation~(\ref{eq:finale}) are $0$. Assume that $D_q>0$.

We prove that for almost all runs $\sigma$ of $\setofruns_j$,
\begin{equation}\label{eq:numer}
  \sum_{s\in S} E_s\cdot D_{s,q} = \lim_{n\rightarrow \infty} \frac{\sum_{i=1}^{n} \ntime{i}(\sigma)\cdot 1^i_q(\sigma)}{n}
\end{equation}
\begin{equation}\label{eq:denom}
  \sum_{p\in Q} \sum_{s\in S} E_s\cdot D_{s,p} = \lim_{n\rightarrow \infty} \frac{\sum_{i=1}^{n} \ntime{i}(\sigma)}{n}
\end{equation}
which proves Equation~(\ref{eq:finale}) because
\[\lim_{n\rightarrow \infty} \frac{\sum_{i=1}^{n} \ntime{i}(\sigma)\cdot
  1^i_q(\sigma)}{ \sum_{i=1}^{n} \ntime{i}(\sigma)}=\cme^\calA_q(\sigma)\,.\]

By the strong law of large numbers, for almost all runs $\sigma$ of $\setofruns$
we have
\begin{equation}\label{eq:slln}
  E_s=\lim_{n \rightarrow \infty} \frac{\sum_{i=1}^{n} \ntime{i}(\sigma) \cdot
    1^i_{s,p}(\sigma)}{\sum_{i=1}^{n} 1^i_{s,p}(\sigma)}
\end{equation}
for all $s\in S$ and $p\in Q$ satisfying $D_{s,p}>0$ (note that waiting
times in $s$ do not depend on $p$).
Let $\sigma$ be a run of $\setofruns$ which satisfies Equation~(\ref{eq:slln})
for all $s\in S$ and $p\in Q$ where $D_{s,p}>0$ and such that $\dme^{\calA
  \times S}$ is well-defined for $\sigma$.

For every $p\in Q$ we have that $\sum_{s\in S} E_s \cdot D_{s,p}$ is equal to
\begin{align*}
 \sum_{s\in S} & \lim_{n \rightarrow \infty} \frac{\sum_{i=1}^{n} \ntime{i}(\sigma) \cdot
  1^i_{s,p}(\sigma)}{\sum_{i=1}^{n} 1^i_{s,p}(\sigma)} \cdot D_{s,p} = \\
   & =  \sum_{s\in S} \lim_{n \rightarrow \infty} \frac{\sum_{i=1}^{n} \ntime{i}(\sigma) \cdot
  1^i_{s,p}(\sigma)}{\sum_{i=1}^{n} 1^i_{s,p}(\sigma)} \cdot
  \frac{\sum_{i=1}^{n} 1^i_{s,p}(\sigma)}{n}\\
  & =  \sum_{s\in S} \lim_{n \rightarrow \infty} \frac{\sum_{i=1}^{n} \ntime{i}(\sigma) \cdot
  1^i_{s,p}(\sigma)}{n} \\
  & =  \lim_{n \rightarrow \infty} \frac{\sum_{i=1}^{n} \ntime{i}(\sigma) \cdot
   \sum_{s\in S} 1^i_{s,p}(\sigma)}{n} \\
   &=  \lim_{n \rightarrow \infty} \frac{\sum_{i=1}^{n} \ntime{i}(\sigma) \cdot
   1^i_{p}(\sigma)}{n} \\
\end{align*}
which proves Equation~(\ref{eq:numer}). Also
$\sum_{p\in Q} \sum_{s\in S} E_s\cdot D_{s,p}$ is equal to
\begin{align*}
\sum_{p\in Q} \lim_{n \rightarrow \infty} \frac{\sum_{i=1}^{n} \ntime{i}(\sigma) \cdot 1^i_{p}(\sigma)}{n}
& = \lim_{n\rightarrow \infty} \frac{\sum_{i=1}^{n} \ntime{i}(\sigma) \cdot \sum_{p\in Q} 1^i_p(\sigma)}{n} \\
& = \lim_{n\rightarrow \infty} \frac{\sum_{i=1}^{n} \ntime{i}(\sigma)}{n}
\end{align*}
which proves Equation~(\ref{eq:denom}) and finishes the proof.
\qed

\section{Proofs of Section~\ref{SEC:PRODUCT-PROCESS}}
\label{app:product}

\subsection{Proof of Lemma~\ref{lem:product-kernel}}
\label{app:product-kernel}

\begin{reflemma}{lem:product-kernel}
  Let $A \in \G_\product$. Then $\kernel_\product(\cdot,A)$ is a
  measurable function, i.e., $\product$ is a GSSMC.
\end{reflemma}
\begin{proof}
  To prove this lemma, it is sufficient to show that
  $\kernel_\product(\cdot,A)$ is a measurable function from
  $\Gamma_\product$ to $[0,1]$ where $A$ ranges (only) over the generators
  of $\G_\product$, i.e. $A=\{s'\}\times\{q'\}\times\vec{I}$ where $s' \in
  S$, $q' \in Q$, and $\vec{I}=\prod_{x \in \calX} I_x$ such that $I_x$ is
  an interval for each $x \in \calX$ (see, e.g., \cite[Lemma~1.37]{Kal:book}).

  As the sets $S$ and $Q$ are finite, our goal is to show that a function
  $P_\product((s,q,\cdot),\{s'\}\times\{q'\}\times\vec{I})$ is measurable for
  $s,s' \in S$, $q,q' \in Q$, and a product of intervals $\vec{I}$.

  The rest of the proof is based on the fact that a real valued function is
  measurable, if it is piecewise continuous. Hence, we finish the proof
  showing that the function
  $P_\product((s,q,\cdot),\{s'\}\times\{q'\}\times\vec{I})$ is piecewise
  continuous when we fix valuation of all clocks but one.  Formally, we fix
  a valuation $\nu$ and a clock $x$ and show that the following function of
  a parameter $u \in \Rsetpo$ is piecewise continuous.
  \begin{multline*}
     P_\product((s,q,\nu[x:=u]),\{s'\}\times\{q'\}\times\vec{I})=
     \delta_{q'\bar{q}} \cdot \prob(s)(s') \cdot \int_0^\infty f(t) \cdot
        1_\vec{I}(\bar{\nu} + t) \de{t}
  \end{multline*}
  where 
  \begin{itemize}
  \item $\nu[x:=u]$ is the valuation $\nu$ where the value of the clock
    $x$ is set to $u$;
  \item $\delta$ is the Kronecker delta, i.e., $\delta_{q'\bar{q}} = 1$ if
    $q'=\bar{q}$, and $0$, otherwise;
  \item $f = \D(s,s')$ is the delay density function for $(s,s')$;
  \item $(\bar{q},\bar{\nu})$ is the timed automaton successor of the state
    $(s,q,\nu[x:=u])$, i.e.,
    $\successor((s,q,\nu[x:=u]))=(\bar{q},\bar{\nu})$;
  \item $1_\vec{I}$ is the indicator function of the set $\vec{I}$, i.e.,
    $1_\vec{I}(\nu') = 1$ if $\nu' \in \vec{I}$, and $0$, otherwise.
  \end{itemize}
  The function $\prob(s)(s')$ is constant (recall that $s$ and $s'$ are
  fixed).  Due to the standard region construction for $\calA$, it holds
  that $\bar{q}$ is piecewise constant and $\bar{\nu}$ is piecewise
  continuous with respect to $u$.

  Let $u$ be in one of the finitely many intervals where
  $\delta_{q'\bar{q}}$ is constant and the valuation $\bar{\nu}$ changes
  continuously, i.e. the automaton $\calA$ uses the same transition for all
  $u$ of this interval. As $1_\vec{I}$ is the indicator function and
  $\vec{I}$ is a product of intervals, it holds that
  \begin{multline*}
     \delta_{q'\bar{q}} \cdot \prob(s)(s') \cdot 
        \int_0^\infty f(t) \cdot 1_\vec{I}(\bar{\nu} + t) \de{t}= 
    \delta_{q'\bar{q}} \cdot \prob(s)(s') \cdot 
        \int_{\mathbf{a}(u)}^{\mathbf{b}(u)} f(t) \de{t}
  \end{multline*}
  where ${\mathbf{a}}(u)$ and ${\mathbf{b}}(u)$ are continuous functions of
  $u$ and so $\int_{\mathbf{a}(u)}^{\mathbf{b}(u)} f(t)\, \de{t}$ is also a
  continous function of $u$ (recall that $\int_{0}^{\infty} f(t)\, \de{t} =
  1$). Therefore, the function
  $P_\product((s,q,\nu[x:=u]),\{s'\}\times\{q'\}\times\vec{I})$ is a
  piecewise continuous function of $u$ and $\kernel_\product(\cdot,A)$ is a
  measurable function, i.e., $\product$ is a GSSMC.
\end{proof}

\subsection{Proof of Lemma~\ref{lem:product-correctness}}
\label{app:product-correctness}

\begin{reflemma}{lem:product-correctness}
  There is a measurable one-to-one mapping $\pr$ from the set of runs
  of $\smp$ to the set of runs of $\product$ such that
  \begin{itemize}
  \item $\pr$ preserves measure, i.e., for every measurable set $X$ 
    of runs of $\smp$ we have that $\pr(X)$ is also measurable and 
    $\probm_\smp(X) = \probm_\product(\pr(X))$;
  \item $\pr$ preserves $\dme$, i.e., for every run $\sigma$ of $\smp$
    and every $q \in Q$ we have that $\dme_{q}(\sigma)$ is well-defined 
    iff $\dme_{q}(\pr(\sigma))$ is well-defined, 
    and $\dme_{q}(\sigma) = \dme_{q}(\pr(\sigma))$.
  \end{itemize}
\end{reflemma}
\begin{proof}
  First, we define the function $\pr$. We use auxiliary functions $\pr_{0}:
  S \rightarrow \Gamma_\product$ that maps the initial states and a function
  $\pr_{\rightarrow}: \Gamma_\product \times (S \times \Rsetp \times S)
  \rightarrow \Gamma_\product$ that maps transitions.  First, we set
  $\pr_0(s) = (s,q_0,\vec{0})$ where $q_0$ is the initial location and
  $\vec{0}$ is the zero vector. Next, let $s,s' \in S$, $t\in\Rsetp$ and $z
  = (s'',q,\nu) \in \Gamma_\product$.  We define
  $\pr_\rightarrow(z,(s,t,s')) = (s',q',\nu'+t)$ such that $\successor(z) =
  (q',\nu')$.

  For a run $\sigma = s_0 t_0 s_1 t_1 \cdots$, we use these two functions to
  set $\pr(\sigma) = z_0 z_1 z_2\cdots$ such that $\pr_{0}(s_0) = z_0$ and
  for each $i \in \Nseto$ it holds that
  $\pr_{\rightarrow}(z_i,(s_i,t_i,s_{i+1})) = z_{i+1}$.  We need to show the
  following claims about the function $\pr$.

  \begin{claim}
    Let $\sigma$ be a run of $\smp$. We have for any $q \in Q$ that
    $\dme_q(\sigma)$ is well-defined if and only if $\dme_q(\pr(\sigma))$ is
    well-defined, and $\dme_q(\sigma) = \dme_q(\pr(\sigma))$.
  \end{claim}

  Let $\sigma = s_0 t_0 s_1 t_1 s_2 t_2 \cdots$ be a run of $\smp$.  Let us
  fix a location $q \in Q$. Recall the Figure~\ref{fig:product}.  The run of
  $\ta$ over $\sigma$ is a sequence
  \begin{align*}
    \ta(\sigma) &= (q_0,\nu_0) s_0 (q_1,\bar{\nu_0}) t_0 (q_1,\nu_1) s_1
    (q_2,\bar{\nu_1}) t_1 (q_2,\nu_2) s_2 \cdots.  
    \intertext{The corresponding run of the product is} 
    \pr(\sigma) &= (s_0,q_0,\nu_0)(s_1,q_1,\nu_1)(s_2,q_2,\nu_2) \cdots.
  \end{align*}
  The values $\dme_q(\sigma)$ and $\dme_q(\pr(\sigma))$ are limit superior
  of partial sums of ratio of $q$ in a sequence of locations. For
  $\dme_q(\sigma)$ the sequence is $Q^1(\sigma), Q^2(\sigma), Q^3(\sigma),
  \ldots = q_1, q_2, q_3, \ldots$ (recall that $Q^i(\sigma)$ is the location
  entered after reading the finite prefix $s_0\,t_0\cdots s_i$) and for
  $\dme_q(\pr(\sigma))$ the sequence is also $q_1, q_2, q_3, \ldots$.
  Hence, we get that $\dme_q(\sigma)$ is well-defined iff
  $\dme_q(\pr(\sigma))$ is well-defined and $\dme_q(\sigma) =
  \dme_q(\pr(\sigma))$.

  \begin{claim}
    For any measurable set $X$ of runs of $\smp$, the set $\pr(X)$ is
    measurable.
  \end{claim}

  Recall that by $\setofruns(B)$ we denote a cylinder of runs that follow
  the given template $B$.  Let $X$ be a set of runs such that $X =
  \setofruns(B)$ for some template $B = s_0 I_0 \cdots s_n I_n$, i.e. $X$ is
  from the generator set.  We can cover the image of $X$ by cylinders
  composed of basic hybercubes. By decreasing the edge length of the
  hypercubes to the limit, we then get a set that equals the image of $X$.
  For $k \in \Nset$ and $\vecv \in (\Nseto)^{|\calX|}$ we denote by
  $C_\vecv^k$ a set of valuations $\prod_{x\in\calX}
  [\vecv(x)/k,(\vecv(x)+1)/k]$. The set of all cylinder templates composed
  of basic hypercubes of precision $k$ is
  \begin{align*}
    U_k &= \{A_0\cdots A_n \mid A_i = \{s_i\} \times \{q_i\} \times C_{\vecv_i}^k, s_i \in S, q_i \in Q, \vecv_i \in (\Nseto)^{|X|}\}
    \intertext{A run $\sigma = z_0 z_1 \cdots$ of $\product$ is in $\setofruns(A_0
      \ldots A_n)$ if for each $0 \leq i \leq n$ we have $z_i \in A_i$.  It
      is easy to show that} 
    \pr(X) &= \bigcap_{k\in\Nset} \bigcup \{\setofruns(C) \mid C \in U_k,
    \setofruns(C) \cap \pr(X) \neq \emptyset\}
  \end{align*}
  hence, $\pr(X)$ is a measurable set. By standard arguments we get the
  result for any measurable $X$.
  \begin{claim}
    For any measurable set $X$ of runs of $\product$, the set $\pr^{-1}(X)$
    is measurable.
  \end{claim}

  The arguments are similar as in the previous claim.  Let $Y$ be a set of
  runs such that $Y = \setofruns(C)$ for some template
\[
 C = \{s_0\}\times\{q_0\}\times \prod_{x\in\calX} I_{x,0} \quad \cdots \quad \{s_n\}\times\{q_n\}\times \prod_{x\in\calX} I_{x,n},
\]
i.e. $Y$ is from the generator set. By $I^k_i$ we denote an interval $[i/k,(i+1)/k]$. The set of all
cylinder templates in $\smp$ composed of basic lines of precision $k$ is
\begin{align*}
T_k &= \{s_0 I^k_{i_0} \cdots s_n I^k_{i_n} \mid s_i \in S, i_n \in \Nseto\}
\intertext{Again, it is easy to show that}
\pr^{-1}(Y) &= \bigcap_{k\in\Nset} \bigcup \{\setofruns(B) \mid  B \in T_k, \setofruns(B) \cap \pr^{-1}(Y) \neq \emptyset\} 
\end{align*}
hence, $\pr^{-1}(Y)$ is a measurable set. Again, by standard arguments we get the result for any measurable $X$.

\begin{claim}
For any measurable set $X$ of runs of $\smp$, we have $\probm_\smp(X) = \probm_\product(\pr(X))$.
\end{claim}

We define a new measure $\probm'_\smp$ over runs of $\smp$ by
\[
 \probm'_\smp(X) = \probm_\product(\pr(X))
\]
for any measurable set of runs $X$. First we need to show that $\probm'_\smp$ is
a probability measure, i.e. $\probm'_\smp(\emptyset) = 0$, $\probm'_\smp(\setofruns_\smp) = 1$, and
for any collection of pairwise disjoint sets $X_1,\ldots,X_n$ we have $\probm'_\smp(\bigcup_{i=1}^n X_i) = \sum_{i=1}^n \probm'_\smp(X_i)$.
The first and the second statement follows directly from the definition of $\pr$, 
the third statement follows from the fact that $\probm_\product$ satisfies this property
and that $\pr$-image of disjoint sets are disjoint sets of runs which can be easily checked.

Let $B = s_0 I_0 \cdots s_n I_n$ be a cylinder template. We show 
\[
 \probm'_\smp(\setofruns(B)) = \probm_\smp(\setofruns(B)).
\]
We obtain $\probm'_\smp = \probm_\smp$ by the extension theorem because $\probm'_\smp$ 
and $\probm_\smp$ coincide on the generators. From the definition of $\probm'_\smp$ we get the claim.
From the definition of semi-Markov process, we have
\begin{align*}
\probm_\smp(\setofruns(B)) &= \distribution_0(s_0) \cdot \prod_{i=0}^{n} \prob(s_i)(s_{i+1}) \cdot \int_{t_i \in I_i} f_i(t_i)\, dt_i \\
\intertext{where $f_i = \D(s_i, s_{i+1})$ is the density of the transition from $s_i$ to $s_{i+1}$. 
Now we turn our attention to the product. For the fixed template we define a set $N_0 = \{(s_0,q_0,\vec{0})\}$ 
and a sequence of functions $N_1,\ldots,N_n$ such that}
N_{i+1}(z_i) &= \{\pr_{\rightarrow}(z_i,(s_i,t_i,s_{i+1})) \mid t_i \in I_i \}.
\end{align*}
For an interval $I = [a,b]$ and valuation $\nu$ we define $C(I,\nu)$ to be a hypercube 
of edge length $b-a$ starting at $\nu+a$, i.e. $C(I,\nu) = \prod_{x\in\calX} [\nu(x)+a,\nu(x)+b]$.
Now for each $i \in \Nseto$ the conditional density
\begin{align*}
 \probm_\product( & \Phi_{i+1} \in N_{i+1}(\Phi_i) \mid \Phi_i) \\ 
 &= \probm_\product( \Phi_{i+1} \in \bar{N}_{i+1}(\Phi_i) \mid \Phi_i)
\intertext{%
where $\bar{N}_{i+1}(z) = \bigcup_{q_{i+1}\in Q} \{s_{i+1}\} \times \{q_{i+1}\} \times C(I_i,\bar{\nu})$ 
where $\ta(z) = (\bar{q},\bar{\nu})$.
The equality holds because $N_{i+1}(\Phi_i) \subseteq \bar{N}_{i+1}(\Phi_i)$ 
and $\probm_\product( \Phi_{i+1} \in (\bar{N}_{i+1}(\Phi_i) \setminus N_{i+1}(\Phi_i)) \mid \Phi_i) = 0$.
Indeed, the probability of hitting anything else but the diagonal of the hypercube $\bar{N}_{i+1}(\Phi_i)$ is clearly $0$.
Because $\bar{N}_{i+1}(z)$ is for each $z$ a union of basic cylinders for that we have explicit definition of the transition kernel,
we have}
 &= \sum_{q\in Q} \delta_{q\bar{q}} \cdot \prob(s_i)(s_{i+1}) \cdot \int_0^\infty f_i(t) \cdot 1_{C(I_i,\bar{\nu})}(\bar{\nu}+t) \de{t} \\
\intertext{%
where $\delta$ is the Kronecker delta, i.e., $\delta_{q\bar{q}} = 1$ if
    $q=\bar{q}$, and $0$, otherwise. Notice that here $\bar{q}$ and $\bar{\nu}$ are random variables
such that $\ta(\Phi_i) = (\bar{q},\bar{\nu})$.
The rest of the formula after $\delta_{q\bar{q}}$ does not depend on $\bar{q}$, we can write}
 &= \prob(s_i)(s_{i+1}) \cdot \int_0^\infty f_i(t) \cdot 1_{C(I_i,\bar{\nu})}(\bar{\nu}+t) \de{t}. \\
\intertext{%
Furthermore, hitting the hypercube $C(I_i,\bar{\nu})$ equals to waiting for a time from $I_i$}
 &= \prob(s_i)(s_{i+1}) \cdot \int_{I_i} f_i(t) \de{t}. 
\end{align*}
Hence, the conditioned probability is a constant random variable that does not depend on $\Phi_i$.
Finally,
\begin{align*}
 \probm'_\smp(\setofruns(B)) &= \probm_\product(\pr(\setofruns(B))) \\
&= \probm_\product(\Phi_0 \in N_0, \Phi_1 \in N_1(\Phi_0),\ldots, \Phi_n \in N_n(\Phi_{n-1})) \\
&= \probm_\product(\Phi_n \in N_n(\Phi_{n-1}) \mid \Phi_0,\ldots,\Phi_{n-1}) \; \cdot \; \; \cdots \\
& \quad \cdot \; \probm_\product(\Phi_1 \in N_n(\Phi_0) \mid \Phi_0) \cdot \probm_\product(\Phi_0 \in N_0) \\
&= \prod_{i=0}^{n} \left( \prob(s_i)(s_{i+1}) \cdot \int_{t_i \in I_i} f_i(t_i)\, dt_i \right) \cdot \distribution_0(s_0) \\
&= \probm_\smp(\setofruns(B))
\end{align*}
which concludes the proof.
\end{proof}

\section{Proofs of Section~\ref{SEC:FINISH-MAIN-THM}}
\label{app:sec-main-thm}
\subsection{Proofs of Proposition~\ref{prop:small-product}}
\begin{refproposition}{prop:small-product}
 Assume that $\mathcal{G}_{\product}$ is strongly connected and aperiodic.
 Then there exists a
region $R$ and $S\subseteq R$, $n\in\Nset$, $b>0$, and a probability measure $\mu$ such that $\mu(S)=1$ and for
every measurable $T\subseteq S$ and $z\in\Gamma_\product$ we have $\kernel_\product^{n}(z,T)>b\cdot\mu(T)$.
In other words, the set $\Gamma_\product$ of all states of the GSSMC $\product$ is $(n,b,\mu)$-small.
\end{refproposition}
\begin{proof}
 Follows easily from Proposition~\ref{prop:main-lem-period} by considering the period $p$ equal to $1$.
\end{proof}

\subsection{Proofs of Proposition~\ref{prop:main-lem-period}}

\begin{refproposition}{prop:main-lem-period}
 Assume that $\mathcal{G}_{\product}$ is strongly connected and has a period $p$.
 For every $k\in \{0,\ldots,p-1\}$ there exists a
set of states $S_k$, $n_k\in\Nset$, $b_k>0$, and a probability measure $\smallmeasure_k$ such that $\smallmeasure_k(S_k)=1$ and for
every measurable $T\subseteq S_k$ and $z\in\Gamma^k_\product$ we have $\kernel_\product^{n_k\cdot p}(z,T)>b_k\cdot\smallmeasure_k(T)$.
In other words, $\Phi_k$ is $(n_k,b_k,\smallmeasure_k)$-small.
\end{refproposition}

In the following text we formulate the definitions and lemmata needed to prove the proposition.
The actual proofs of the lemmata are in next subsections (grouped by proof techniques).

Let us fix a $k\in \{0,\ldots,p-1\}$. We show that there is a state $z^\ast \in\Gamma^k_\product$ such that for each starting state
$z\in\Gamma^k_\product$ there is a path $z\cdots z^\ast$ of length $n_k\cdot p$ that is \emph{$\delta$-wide}.
For a fixed $\delta>0$, it means that the waiting time of any transition in the path
can be changed by $\pm\delta$ without ending up in a different region in the end.
Precise definition follows.

\begin{definition}
Let $z = (s,q,\nu)$ and $z' = (s',q',\nu')$ be two states. For a waiting time $t \in \Rsetp$
we set $\trans{z}{t}{z'}$ if $\successor(z) = (q',\bar{\nu})$ and $\nu' = \bar{\nu} + t$.
We set $\trans{z}{}{z'}$ if $\trans{z}{t}{z'}$ for some $t\in\Rsetp$ and call it
a \emph{feasible transition}.

For $\delta>0$, we say that a feasible transition $\trans{z}{}{z'}$ is $\delta$-wide
if for every $x \in \calX$ relevant for $\nu'$ we have
$\fr(\nu_i(x)) \in [\delta,1-\delta]$.

Let $z_1 \cdots z_n$ be a path.
It is \emph{feasible} if for each $1 \leq i < n$ we have that $\trans{z_i}{}{z_{i+1}}$.
It is \emph{$\delta$-wide} if for each $1 \leq i < n$ we have that $\trans{z_i}{}{z_{i+1}}$ is a $\delta$-wide transition.
\end{definition}

By next lemma, we reduce the proof of Proposition~\ref{prop:main-lem-period} to finding
$\delta$-wide paths from any $z$ to the fixed $z^\ast$.

First, we recall the following notation that is necessary for analyzing the computational complexity.
Let $\pmin$ denote the smallest probability in $\smp$. Further, let us denote by
$\delay(\smp)$ the set of delay densities used in $\smp$, i.e.~$\delay(\smp)=\{\D(s,s')\mid s,s'\in S\}$.
From our assumptions imposed on delay densities we obtain the following uniform
bound $\densb>0$ on delay densities of $\delay(\smp)$. For every $f\in\delay(\smp)$
and for all $x\in[0,\maxb]$, either $f(x)>\densb$ or $f(x)=0$, and moreover,
$\int_{\maxb}^\infty f(x)dx>c$ or equals $0$.

\begin{lemma}\label{lem:fuzzying}
For every $\delta > 0$ and $n>1$ there is a probabilistic measure $\smallmeasure$
and $b >0$ such that the following holds.
For every $\delta$-wide path
$\sigma = z_0z_1 \cdots z_n$, there is a $\smallmeasure$-measurable set of states
$Z$ with $\smallmeasure(Z)=1$ such that $z_n \in Z$ and for any measurable subset $Y
\subseteq Z$ it holds $\kernel_\product^n(z_1,Y) \geq b \cdot \smallmeasure(Y)$.

Moreover, we can set $b=(\pmin \cdot \densb \cdot \delta/n)^n/\sqrt{|\calX|}%
$.
\end{lemma}

Now it remains to find a state $z^\ast$ and a $\delta$-wide path to $z^\ast$ for any $z$.
Such path is composed of five parts, each having a fixed length.
The target state $z^\ast$ is then the first state where all
these paths from all starting states $z$ meet together.

In the first part, we move to a \emph{$\delta'$-separated} state for some $\delta' > 0$.
\begin{definition}\label{def:delta-separation}
  Let $\delta > 0$. We say that a set $X \subseteq \Rsetpo$ is
  $\delta$-separated if for every $x,y\in X$ either $\fr(x) =
  \fr(y)$ or $|\fr(x) - \fr(y)| > \delta$.

Further, we say that $(s,q,\nu)\in\Gamma_\product$ is
$\delta$-separated if the set
  $$
    \{ 0 \} \cup \{ \nu(x) \mid x \in \calX, x \textrm{ is relevant
for }\nu \}
  $$
  is $\delta$-separated.
\end{definition}

Now we can formulate the first part of the path precisely.

\begin{lemma}\label{lem:separation}
There is $\delta > 0$ and $n \in \Nset$ such that for any $z_1 \in
\Gamma_\product$ there is a $\delta$-wide path $z_1 \cdots z_n$ such
that $z_n$ is $\delta$-separated.

\noindent
Moreover, we can set $n = \maxb \cdot |\calX|$ and $\delta = 1 / ( 2 (|\calX| + 2))$.
\end{lemma}

At the beginning of the second part, we are in a $\delta$-separated state $z_1$
in some region $R \in V_{k'}$ for some $k' \in \{0,\ldots,p-1\}$.
For the given $k'$, we fix a region $R_1 \in V_{k'}$. Due to strong connectedness, reaching $R_1$ is possible
from any state in $V_{k'}$ in a fixed sufficiently large number of steps $n'$.
By a path of length $n'$ that is $(\delta/n')$-wide, we reach a $(\delta/n')$-separated state in $R_1$.
The separation and wideness decreases with length because the fractional values
of relevant clock may be forced to get closer and closer to each other by resets on the path to $R_1$.
The reason for the first part of the path was only to bound the wideness of the second part.

\begin{lemma}\label{lem:path}
Let the region graph $\regiongraph_\product$ be strongly connected and let $p$ be the period of $\regiongraph_\product$.
Let $k\in\{0,\ldots,p-1\}$, $\delta > 0$ and $R \in V_k$ be a region. Then there is $n \in \Nset$ such that for
every $\delta$-separated $z_1 \in \Gamma^k_\product$ there is
a $(\delta/n)$-wide path $z_1 \cdots z_n$ such that $z_n$ is $(\delta/n)$-separated and $z_n \in R$.

\noindent Moreover, we can set $n= \lfloor
\sizereg^{4\ln\sizereg-1}/6 \rfloor \cdot p$.
\end{lemma}

In the third part, we squeeze the fractional values of all relevant clocks close to $0$.
We go through a fixed region path $R_1 \cdots R_k$ such that in each step
we shift the time by an integral value minus a small constant $c$.
This way the reset clocks are fractionally placed to $0$
and the other clocks decrease their fractional values only by the small constant $c$.
Since we go through a fixed region path, we have a fixed sequence of sets of clocks $\calX_1,\ldots,\calX_{k-1}$
reset in respective steps. Hence, the fractional values of clocks reset during this path have fixed relative distances.
For any starting state $z'_1$ we reach a state $z'_k$ that \emph{almost equals} a fixed ``reference'' state $z_k \in R_k$.
\begin{definition}
 Let $z,z' \in \Gamma_\product$. We say that state $z$ \emph{almost equals} state $z'$ if $z\region z'$ and
each clock relevant in $z$ has the same value in $z$ and $z'$.
\end{definition}
Notice that clocks not relevant in $z_k$ may still have different values.
\begin{lemma}\label{lem:corner}
Let $R$ be a region. For each $\delta > 0$ there is $\delta' > 0$, $n \in \Nset$ and
$z' \in \Gamma_\product$ such that for every $\delta$-separated $z_1 \in R$ there is a $\delta'$-wide path $z_1 \cdots z_n$ such that
$z_n$ and $\bar{z}$ almost equal.

\noindent
Moreover, we can set $n = \maxb + 1$ and $\delta' = \delta / (\maxb+2)$.
\end{lemma}

In the fourth part, we somewhat repeat the first part and prepare for the fifth part.
We reach a $\delta$-separated state that is almost equal to a fixed state $z_l \in R_l$.
Again, we do it by a $\delta$-wide path.
\begin{lemma}\label{lem:almost-equal-separation}
Let $z$ be a state. There is a $\delta > 0$, $n\in\Nseto$, and $z'$
such that for any state $z_1$ almost equal to $z$ there is a $\delta$-wide path
$z_1 \cdots z_n$ such that $z_n$ is $\delta$-separated and $z_n$ almost equals $z'$.

\noindent
Moreover, we can set $n = \maxb \cdot |\calX|$ and $\delta = 1 / ( 2 (|\calX| + 2))$.
\end{lemma}

In the fifth part, we go through a fixed region path $R_l \cdots R_{l+m}$
such that each clock not relevant in $R_l$ is reset during this path and
hence we reach a fixed $z^\ast \in R_{l+m}$.
Such path exists from the assumption that it is possible to reset every clock.
The (arbitrary) values of clocks not relevant in $R_l$ do not influence the behavior
of the timed automaton before their reset and we indeed follow a fixed region path.
Furthermore, we can stretch the path to arbitrary length so that the length
of the whole path is a multiple of the period $p$.
Again, the fifth part of the path is $\delta/n''$-wide where $n''$ is the number of steps.
\begin{lemma}\label{lem:synchronization}
Let the region graph $\regiongraph_\product$ be strongly connected.
Let $\delta > 0$. Let $z$ be a $\delta$-separated state. Then there is $n\in\Nset$, such
that for any $n' \geq n$ there is a state $z^\ast$ such that the following holds.
For any state $z_1$ almost equal to $z$ there is a $(\delta/n)$-wide path
$z_1 \cdots z_{n'}$ such that $z_{n'} = z^\ast$.

\noindent
Moreover, we can set $n = \sizereg \cdot |\calX|$.
\end{lemma}

Now we can finally prove the main proposition.

\begin{proof}[of Proposition~\ref{prop:main-lem-period}]
We fix $k\in \{0,\ldots,p-1\}$. By Lemmata~\ref{lem:separation},
\ref{lem:path}, \ref{lem:corner}, \ref{lem:almost-equal-separation},
and \ref{lem:synchronization} we get for any state
$z_1\in\Gamma^k_\product$ a $\delta$-wide path $z_1 \cdots z_x$ of
length $x = n_k \cdot p$ such that
\begin{align*}
 x &= \maxb \cdot |\calX| + M + (\maxb + 1) + \maxb \cdot |\calX| + \sizereg\cdot|\calX| + c \leq 2 \cdot M \\
 \delta &= \frac{1}{((\maxb+2) \cdot 2(|\calX|+2) \cdot M} \geq \frac{1}{4 \cdot \maxb \cdot |\calX| \cdot M}
\end{align*}
where $M = \lfloor (\sizereg^{4 \ln \sizereg-1})/6 \rfloor \cdot
p$ and $c < p$ is the constant such that $x$ is a multiple of $p$
(we stretch the path by $c$ in the fifth part). Therefore, by
Lemma~\ref{lem:fuzzying}, $\Phi_k$ is
$(n_k,b_k,\smallmeasure_k)$-small for $n_k \leq \lfloor \sizereg^{4
\ln \sizereg-1} \rfloor \cdot p \leq \lfloor \sizereg^{4
\ln \sizereg} \rfloor=:r$ and
\begin{align*}
b_k &= (\pmin \cdot \densb \cdot \delta/x)^x/\sqrt{|\calX|}
\\
  &\geq \left( \frac{\pmin \cdot \densb}{8 \cdot \maxb \cdot |\calX| \cdot M^2} \right)^{2M} \cdot \frac{1}{\sqrt{|\calX|}}
\\
  &\geq \left( \frac{\pmin \cdot \densb}{M^3} \right) ^{2M}
  \geq \left( \frac{\pmin \cdot \densb}{%
  r} \right) ^{r}
\end{align*}
Notice that in the calculation above we ignore the cases of
trivially small region graphs with less than $3$ vertices.
\end{proof}

From the proof we directly get the following bound on constants.

\begin{corollary}
 Assume that $\mathcal{G}_{\product}$ is strongly connected and has a period $p$.
 For every $k\in \{0,\ldots,p-1\}$ we have $\Phi_k$ is $(n,b,\smallmeasure)$-small
 for some $n \leq r$ divisible by $p$,
and $b = ( \pmin \cdot \densb/ r ) ^{r}$, where $r = \lfloor \sizereg^{4 \ln \sizereg} \rfloor$.
\end{corollary}

\subsubsection{Proof of Lemma~\ref{lem:fuzzying}}

\begin{reflemma}{lem:fuzzying}
For every $\delta > 0$ and $n>1$ there is a probabilistic measure $\smallmeasure$
and $b >0$ such that the following holds.
For every $\delta$-wide path
$\sigma = z_0z_1 \cdots z_n$, there is a $\smallmeasure$-measurable set of states
$Z$ with $\smallmeasure(Z)=1$ such that $z_n \in Z$ and for any measurable subset $Y
\subseteq Z$ it holds $\kernel_\product^n(z_1,Y) \geq b \cdot \smallmeasure(Y)$.

Moreover, we can set $b=(\pmin \cdot \densb \cdot \delta/n)^n/\sqrt{|\calX|}%
$.
\end{reflemma}

\begin{proof}
Recall that we assume that all delays' densities are bounded by some
$\densb>0$ in the following sense. For every $d\in\delay$ %
and for all $x\in[0,B]$, $d(x)>\densb$ or equals $0$. Similarly, $\int_{B}^\infty d(x)dx>\densb$ or equals $0$.

Let $\sigma=z_0z_1\cdots
z_n=(s_0,q_0,\nu_0)(s_1,q_1,\nu_1)\cdots(s_n,q_n,\nu_n)$. For $1\leq
i\leq n$, let $t_i$ be the waiting times such that
$\trans{z_{i-1}}{t_i}{z_{i}}$, and let $X_i=\{x\in\calX\mid
\successor(z_{i-1}) = (q,\nu),\; \nu(x)=0\}$ be the set of clocks reset right before waiting $t_i$.

For $\varepsilon>0$, we define an $\varepsilon$-neighbourhood of
$\sigma$ to be the set of paths of the form
$\trans{z_0}{t'_1}(s_1,q_1,\nu'_1)\cdots\trans{}{t'_n}{(s_n,q_n,\nu'_n)}$
where $t'_i\in (t_i-\varepsilon,t+\varepsilon)$. Due to
$\delta$-separation of $\sigma$, all paths of its
$\delta/n$-neighbourhood are feasible. Considering this
$\delta/n$-neighbourhood, the set of all possible $\nu'_n$s forms
the sought set $Z$. We may compute this set as follows. We define a
mapping $\alpha_\sigma:(-\varepsilon,\varepsilon)^{n}\to \Rsetpo^{|\calX|}$
so that $\alpha_\sigma(\zeta_1,\ldots,\zeta_n)=\nu'_n$ for $t'_i=t_i+\zeta_i$. This can
be done by setting
$\alpha_\sigma(\zeta_1,\ldots,\zeta_n)(x)=\sum_{r_x}^n(t_i+\zeta_i)$, where the clock
$x$ was reset in the $r_x$th step for the last time in $\sigma$, i.e.~$r_x=\max\{i\mid x\in X_i\}$.
Obviously, $\alpha_\sigma$ is a restriction of a linear mapping. Therefore,
$\alpha_\sigma((-\varepsilon,\varepsilon)^{n})$ is an open
rhombic hypercube of a dimension $1\leq d\leq |\calX|$.
Due to the last summand, it has a positive
$\smallmeasure_d$-measure. (Here $\smallmeasure_d$ is the standard
Lebesgue measure on the $d$-dimensional affine space that contains
$\alpha_\sigma((-\varepsilon,\varepsilon)^{n})$. Equivalently, it is
the $d$-dimensional Hausdorff measure multiplied by the volume of
unit $d$-ball.)

We set $Z:=\alpha_\sigma((-\delta/2n,\delta/2n)^n)$. Thus, for every $z\in Z$ there is a $\delta/2$-separated path $\tau$ from $z_0$ to $z$. We need to construct $b>0$ such that for all $Y\subseteq Z$, we have $\kernel_\product^n(z_0,Y)\geq b \smallmeasure_d(Y)/\smallmeasure_d(Z)=:b\smallmeasure(Y)$. It is sufficient to prove this for generators of the same topology. We pick the generators as follows. For $z\in Z$ and $\varepsilon<\delta/2n$, we denote $Y(z,\varepsilon)=\alpha_\sigma((-\delta/n,\delta/n)^n)\cap C^{|\calX|}_{z,\varepsilon}$, where $C^{|\calX|}_{z,\varepsilon}$ is a hypercube with dimension $|\calX|$ and size $\varepsilon$ centered in $z$. Clearly, the set of all $Y(z,\varepsilon)\subseteq Z$ form a generator set. We now construct $b>0$ so that for every such $Y:=Y(z,\varepsilon)$ we have
$\kernel_\product^n(z_0,Y)\geq b \smallmeasure_d(Y)/\smallmeasure_d(Z)$. To this end, we prove later on that
\begin{equation}\label{eq:rozpliz}
\kernel_\product^n(z_0,Y)\quad\geq\quad (\pmin\densb/n)^n\delta^{n-d}\varepsilon^d
\end{equation}
Since $\smallmeasure_d(Y)\leq \sqrt{|\calX|}%
\cdot (\varepsilon)^d$ and $\smallmeasure_d(Z)\geq (\delta/n)^d$, we can set $b=(\pmin\densb\delta/n)^n/\sqrt{|\calX|}%
$.

It remains to prove (\ref{eq:rozpliz}).
Let $k_1,\ldots,k_d$ be the elements of $\{r_x\mid x\in\calX\}$ in the increasing order, and $\ell_1,\ldots,\ell_{n-d}$ the remaining numbers in $\{1,\ldots,n\}$. Note that since $\alpha$ is linear, $\alpha^{-1}(Y)$ is $\lambda_n$-measurable ($\lambda_n$ denotes the standard Lebesgue measure on $\Rset^n$). Intuitively, if we want to make clock $x$ hit $Y$, it is sufficient to adjust the waiting time after the last reset of $x$. Let $Y|X_i$ denote the projection of $Y$ to coordinates in $X_i$ (setting other components to zero). The first equation makes use of the facts that (1) all components of each point of $Y|X_i$ have the same value (because $Y$ is a subset of image of $\alpha_\sigma$) and (2) when factoring out all (identical) components but one of each $X_i$, the image of $Y$ is a $d$-hypercube (due to the intersection with $C^{|\calX|}_{z,\varepsilon}$), so we can use projections independently.

\begin{align*}
\kernel_\product^n(z_0,Y) \;\; \geq & \;\; \int_{\{\vec{x}\mid\alpha_\tau(\vec{x})\in Y\}} (\pmin\densb)^n d\lambda_n\ \\=&\;\; (\pmin\densb)^n
\int_{t_{\ell_1}-\delta/2n}^{t_{\ell_1}+\delta/2n} \cdots
\int_{t_{\ell_{n-d}}-\delta/2n}^{t_{\ell_{n-d}}+\delta/2n} \\&\;\;
\int_{\alpha_\tau(0,\ldots,0,d_{k_d},\ldots,d_n)|X_{k_d}\in Y|X_{k_d}} \cdots
\int_{\alpha_\tau(0,\ldots,0,d_{k_1},\ldots,d_n)|X_{k_1}\in Y|X_{k_1}} \\&\;\;
d\zeta_{k_1}\cdots d\zeta_{k_d}d\zeta_{\ell_{n-d}}\cdots d\zeta_{\ell_1}\qquad\quad\ \\=&\;\; (\pmin\densb)^n
\underbrace{\int_{-\delta/2n}^{\delta/2n} \cdots
\int_{-\delta/2n}^{\delta/2n}}_{n-d}
\underbrace{\int_{-\varepsilon/2n}^{\varepsilon/2n} \cdots
\int_{-\varepsilon/2n}^{\varepsilon/2n}}_{d} \\& \;\;
d\zeta_{k_1}\cdots d\zeta_{k_d}d\zeta_{\ell_{n-d}}\cdots d\zeta_{\ell_1}\qquad\quad\ \\= &\;\;
(\pmin\densb)^n(\delta/n)^{n-d}(\varepsilon/n)^d= (\pmin\densb/n)^n\delta^{n-d}\varepsilon^d \qed
\end{align*}
\end{proof}

\subsubsection{Proofs of Lemmata~\ref{lem:separation} and \ref{lem:almost-equal-separation}}

\begin{reflemma}{lem:separation}
There is $\delta > 0$ and $n \in \Nset$ such that for any $z_1 \in
\Gamma_\product$ there is a $\delta$-wide path $z_1 \cdots z_n$ such
that $z_n$ is $\delta$-separated.

\noindent
Moreover, we can set $n = \maxb \cdot (|\calX| + 2)$ and $\delta = 1 / ( 2 (|\calX| + 2))$.
\end{reflemma}

\begin{proof}
To simplify the argumentation we introduce a notion of a \emph{$r$-grid} that marks $r$ distinguished
points (called \emph{lines}) on the $[0,1]$ line segment. In the proof we show that we can place
fractional values of all relevant clocks on such distinguished points.
Let $r \in \Nset$.
We say that a set of clocks $\calY \subseteq \calX$ is \emph{on $r$-grid} in $z$ if
for every $x \in \calY$ relevant in $z$ we have $\fr(\nu(x)) = n/r$ for some $0 \leq n < r$.
For $0 \leq n < r$, we say that the \emph{$n$-th line of the $r$-grid is free} in $z$ if there
is no relevant clock in the $1/2k$-neighborhood of the $n$-th line,
i.e. for any relevant $x\in\calX$ we have $\fr(\nu(x)) \not\in (n/r - 1/2r,n/r + 1/2r)$.

Let $r = |X|+2$. We inductively build a $1/2r$-wide path $z_1\cdots z_n$ where $n = \maxb \cdot r$.
The set $\emptyset$ is on $r$-grid in $z_1$. We show that if a set $\calY_i$
is on $r$-grid in state $z_i$, there is
a $1/2k$-wide transition to $z_{i+1}$ such that $(\calY_i \cup \mathcal{Z})$ is on $r$-grid in $z_{i+1}$
where $\mathcal{Z}$ is the set of clocks newly reset in $z_i$.
There are $|X|+2$ lines on the grid and only $|X|$ clocks. At least two of these lines must be free.
Let $j \neq 0$ be such a line. Let $t$ be a waiting time and $z_{i+1}$ a state
such that $\fr(t) = 1 - j/r$ and $\trans{z_i}{t}{z_{i+1}}$. Such waiting time must be indeed possible
because the interval where the density function of any transition is positive has integral bounds.
The transition $\trans{z_i}{t}{z_{i+1}}$ is $1/2r$-wide because the line $j$ is free in $z_i$.
Furthermore, the set $(\calY_i \cup \mathcal{Z})$ is on $r$-grid in $z_{i+1}$ because the fractional value of
each clock that was previously on $r$-grid was changed by a multiple of $1/r$. The newly reset clocks have
fractional value $1 - j/r$ which is again a multiple of $1/r$.

Next, we show that $\calX$ is on $r$-grid in $z_n$. Clocks reset in this path on $r$-grid in $z_n$.
The remaining clocks are all irrelevant because the path of $\maxb \cdot r$ steps takes at least $\maxb$ time units.
Indeed, each transition in this path takes at least $1/r$ time unit.
According to the definition, $\calX$ is on $r$-grid in $z_n$.
Hence, the state $z_n$ is $1/r$-separated because the distance between two adjacent grid lines is $1/r$.
By setting $\delta = 1/2r$ we get the result.
\end{proof}

\begin{reflemma}{lem:almost-equal-separation}
Let $z$ be a state. There is a $\delta > 0$, $n\in\Nseto$, and $z'$
such that for any state $z_1$ almost equal to $z$ there is a $\delta$-wide path
$z_1 \cdots z_n$ such that $z_n$ is $\delta$-separated and $z_n$ almost equals $z'$.

\noindent
Moreover, we can set $n = \maxb \cdot |\calX|$ and $\delta = 1 / ( 2 (|\calX| + 2))$.
\end{reflemma}
\begin{proof}
Let us fix a state $z_1$ almost equal to $z$. By Lemma~\ref{lem:separation} we get
a $\delta$-wide path $z_1 \ldots z_n$ such that $z_n$ is $\delta$-separated.

Notice that for a fixed state $z$, control state $s$ and time $t$ there is a unique
location $q$ and valuation $\nu$, hence a unique state $z' = (s,q,\nu)$ such that $\trans{z}{t}{z'}$.

Let $t_1,\ldots,t_{n-1}$ be the waiting times and $s_1,\ldots,s_n$ the control states
on the path $z_1,\ldots,z_n$.
For any $\bar{z}_1$ almost equal to $z_1$ we can build using the same waiting times and control states
a path $\bar{z}_1 \cdots \bar{z}_n$.
It is easy to see that for two almost equal states $z,\bar{z}$ a control state $s$ and a time $t>0$ the
states $z',\bar{z'}$ determined by $s$ and $t$ are also almost equal.
Inductively, we get that $\bar{z}_n$ is almost equal to $z_n$.
It also holds that $\bar{z}_n$ is $\delta$-separated since $\delta$-separation is defined
only with respect to relevant clocks.
\end{proof}

\subsubsection{Proofs of Lemmata~\ref{lem:path} and \ref{lem:synchronization}}

For the proof of Lemma~\ref{lem:path} we need the following result from graph theory.

\begin{lemma}\label{lem:ergo-paths}
Let $G$ be a strongly connected and aperiodic oriented graph with
$N>2$ vertices. Then for each $n \geq \lfloor N^{4\ln N-1}/6 \rfloor$,
there is a path of length precisely $n$ between any two vertices of
$G$.
\end{lemma}
\begin{proof}
It is a standard result from the theory of Markov chains, see
e.g.~\cite[Lemma~8.3.9]{Rosenthal:book}, that in every
ergodic Markov chain there is $n_0$ such that between any two states
there is a path of any length greater than $n_0$. In the following,
we give a simple bound on $n_0$.

Let $u,v$ be vertices. By aperiodicity, there are $C$ cycles on $u$
of lengths $c_1,\ldots,c_C\leq N$ with $\gcd(c_1,\ldots,c_C)=1$.
Thus by B\'ezout's identity, there are $m_i\in\Nseto$ such that
$1=\sum_{i=1}^C m_ic_i$. Hence also $1=\sum_{i=1}^C
(m_i+k_i/c_i\cdot\prod_{j=1}^Cc_j)c_i$ for any $k_1+\cdots +k_C=0$.
Therefore, $1=\sum_{i=1}^j n_ic_i$ with some
$0>n_i>-1/c_i\cdot\prod_{j=1}^C c_j$ for $i<C$ and $n_C>0$.
By~\cite[Theorem A.1.1]{Bremaud:book}, $n_0$ can be chosen $N+P(P-1)$,
where $P=\sum_{i=1}^{C-1} |n_i|c_i$, i.e.~the absolute value of the
negative part of the sum. Note that $P<(C-1)N^{C}$.

Let $c_1$ have $F$ different prime factors. Then $c_2$ can be chosen
indivisible by some of the factors. Then $c_3$ can be chosen
indivisible by some of the remaining factors and so on. Therefore,
we can choose $c_i$ so that $C\leq F+1$.
By~\cite[V.15.1.b]{MSC:book-number-theory}, %
for the number $\omega(N)$ of distinct prime factors of $N$, we have
$F\leq \omega(N)<1.39\ln N/\ln\ln N$. Hence $P<1.39\ln N/\ln\ln N
\cdot N^{1+1.39\ln N/\ln\ln N}$ and thus $n_0<N^{4\ln N-1}/6$.

\end{proof}

Both proofs of Lemmata~\ref{lem:path} and \ref{lem:synchronization} use a technique
expressed by the next lemma.

\begin{lemma}\label{lem:wide-path}
Let $\delta > 0$, $n \in \Nset$, $z_1$ be a $\delta$-separated state and $z_1 z'_2 \cdots z'_n$ be a feasible path.
Then there is a $(\delta/n)$-wide path $z_1 z_2 \cdots z_n$ such that $z_n$ is $(\delta/n)$-separated and
for each $1 \leq i \leq n$ we have $z_i \region z'_i$.
\end{lemma}
\begin{proof}
For simplicity, we first transform this path into a
$\delta/2^n$-wide one. We then show how
to improve the result to $\delta/n$-wideness.

For $j\leq n$, we successively construct paths $z_1\cdots z_j$ that are
$\delta/2^j$-wide and $z_j$ is in the same region as $z_j'$ and now is
also $(\delta/2^j)$-separated. The state $z_1$ satisfies all requirements as $z_1$ is $\delta$-separated.
Let $z_1\cdots z_j$ satisfy the requirements. In particular,
$z_j=(s_j,q_j,\nu_j)$ is in the same region as $z_j'=(s_j,q_j,\nu_j')$. Since
there is a waiting time $t'$ with
$\trans{z_j'}{t'}{z_{j+1}'}=(s_{j+1},q_{j+1},\nu_{j+1}')$, there is also an
interval of waiting times $(a,b)$
such that for every $t\in(a,b)$ we end up in the same region,
i.e.~$\trans{z_j}{t}{z}$ for some $z$ of the region containing
$z_{j+1}'$. Moreover, due to
$\delta/2^j$-separation of $\nu_j$, we obtain $b-a\geq \delta/2^j$. Therefore,
we can choose the waiting time $t=a+\text{\textonehalf}\cdot\delta/2^j$ so
that also $\nu_{j}+t$ is
$\delta/2^{j+1}$-separated. Hence also
$\nu_{j+1}:=(\nu_j+t)[\{x\mid\nu_{j+1}'(x)=0\}:=0]$ is $\delta/2^{j+1}$-separated. We
set $z_{j+1}:=(s_{j+1},q_{1+1},\nu_{j+1})$.

Notice, that this approach guarantees that the fractional parts of the just
reset clock are ``in the middle'' between the surrounding clocks. That is why
we needed exponential, i.e.~$2^n$, deminution of the separation. Nevertheless,
due to $\delta$-separation, for every $x,y\in\calX$ there are at least $n$
values between $\fr(\nu(x))$ and $\fr(\nu(y))$ such that even if all were
fractional values of other clocks, the state would be $\delta/n$-separated.
Also note that as the path is only $n$ steps long, there can be at most $n$
different clocks set between any two clocks. Since we know their ordering in
advance, these $n$ different positions are sufficient.
\end{proof}

Now, we can finally start with the promised proofs. Lemma~\ref{lem:path} is a
corollary of Lemmata~\ref{lem:ergo-paths} and \ref{lem:wide-path}.

\begin{reflemma}{lem:path}
Let the region graph $\regiongraph_\product$ be strongly connected and let $p$ be the period of $\regiongraph_\product$.
Let $k\in\{0,\ldots,p-1\}$, $\delta > 0$ and $R \in V_k$ be a region. Then there is $n \in \Nset$ such that for
every $\delta$-separated $z_1 \in \Gamma^k_\product$ there is
a $(\delta/n)$-wide path $z_1 \cdots z_n$ such that $z_n$ is $(\delta/n)$-separated and $z_n \in R$.

\noindent Moreover, we can set $n= \lfloor
\sizereg^{4\ln\sizereg-1}/6 \rfloor \cdot p$.
\end{reflemma}

\begin{proof}
In the region graph we have a partition of vertices to sets $V_0,\ldots,V_{p-1}$ due to Lemma~\ref{lem:periodic-graph}.
Let us fix a $k\in\{0,\ldots,p-1\}$. We can define an aperiodic oriented graph
$(V_k,E_k)$ where $(R,R') \in E_k$ if there is a path from the region $R$ to the region $R'$
of length exactly $p$ in the region graph $\regiongraph_\product$.

Let us fix $\delta > 0$ and a region $R \in V_k$. Due to the strong
connectedness and aperiodicity of $(V_k,E_k)$ we have by
Lemma~\ref{lem:ergo-paths} in the graph $(V_k,E_k)$ from any region
$R' \in V_k$ a path to $R$ of length $x  = \lfloor \sizereg^{4\ln
\sizereg-1}/6 \rfloor > \lfloor |V_k|^{4\ln |V_k|-1}/6 \rfloor$.
Hence, in the graph $\regiongraph_\product$, we have from $R'$ to
$R$ a path of length $n = x \cdot p$.

For every $z_1 \in \Gamma^k_\product$ we have a feasible path
$z_1z_2'\cdots z_n'$ with $z_n'\in R$.
We get the $(\delta/n)$-wide path by applying Lemma~\ref{lem:wide-path}.
\end{proof}

\begin{reflemma}{lem:synchronization}
Let the region graph $\regiongraph_\product$ be strongly connected.
Let $\delta > 0$. Let $z$ be a $\delta$-separated state. Then there is $n\in\Nset$, such
that for any $n' \geq n$ there is a state $z^\ast$ such that the following holds.
For any state $z_1$ almost equal to $z$ there is a $(\delta/n)$-wide path
$z_1 \cdots z_{n'}$ such that $z_{n'} = z^\ast$.

\noindent
Moreover, we can set $n = \sizereg \cdot |\calX|$.
\end{reflemma}
\begin{proof}
Let $\mathcal{Z}$ be the set of clocks that are not relevant in $z$.
For each clock $x \in \mathcal{Z}$ there is a region $R_x$ such that clock $x$ is reset in region $R_x$
(we make this assumption in Section~\ref{SEC:FINISH-MAIN-THM}). Let us fix a state $z_1$ almost equal to $z$.
From the strong connectedness we get a feasible path $z_1 z'_2 \cdots z'_n$
that for each $x \in \mathcal{Z}$ visits the region $R_x$.
Furthermore, $n \leq \sizereg \cdot |\mathcal{Z}| \leq \sizereg \cdot |\calX|$.
From Lemma~\ref{lem:wide-path} we get a $(\delta/n)$-wide path $z_1 z_2 \cdots z_n$ that also
for each $x \in \mathcal{Z}$ visits the region $R_x$.

Notice that for a fixed state $z$, control state $s$ and time $t$ there is a unique
location $q$ and valuation $\nu$, hence a unique state $z' = (s,q,\nu)$ such that $\trans{z}{t}{z'}$.

Let $t_1,\ldots,t_{n-1}$ be the waiting times and $s_1,\ldots,s_n$ the control states
on the path $z_1,\ldots,z_n$.
For any $\bar{z}_1$ almost equal to $z_1$ we can build using the same waiting times and control states
a path $\bar{z}_1 \cdots \bar{z}_n$.
It is easy to see that for two almost equal states $z,\bar{z}$ a control state $s$ and a time $t>0$ the
states $z',\bar{z'}$ determined by $s$ and $t$ are also almost equal.
Inductively, we get that $\bar{z}_i$ is almost equal to $z_i$ for each $1 \leq i \leq n$.
Hence, the path $\bar{z}_1 \cdots \bar{z}_n$ is also
$(\delta/n)$-wide because $\delta$-wideness is defined only with respect to relevant clocks.
We show that $z_n = \bar{z}_n ( = z^\ast)$.

We need a parametrized version of almost equality.
For a set of clocks $\calY$ and two states $z = (s,q,\nu)$ and
$\bar{z} = (\bar{s},\bar{q},\bar{\nu})$ we say that they are $\calY$-equal if
$z \region \bar{z}$ and for each $x\in\calY$ we have $\nu(x) = \bar{\nu}(x)$.
The states $z_1$ and $\bar{z}_1$ are $\calX_1$-equal where $\calX_1 = \calX \setminus \mathcal{Z}$.
Let $\calX_i$ be a set of clocks and $z_i$ and $\bar{z}_i$ be $\calX_i$-equal states.
For any $t > 0$ and two states $z_{i+1}$ and $\bar{z}_{i+1}$ such that $\trans{z_i}{t}{z_{i+1}}$ and
$\trans{\bar{z}_i}{t}{\bar{z}_{i+1}}$ we have $z_{i+1}$ and $\bar{z}_{i+1}$ are $(\calX_i \cup \calY)$-equal
where $\calY$ is the set of clocks reset in $z_i$.
We get that $z_n$ and $\bar{z}_n$ are $\calX$-equal, i.e. $z_n = \bar{z}_n$.
It holds because all clocks from $\mathcal{Z}$ are reset on the path $z_1 \cdots z_n$.

Now, for arbitrary $n' \geq n$, we can stretch the path to $z_1 \cdots z_n \cdots z_{n'}$.
We get $z_{n'} = \bar{z}_{n'}$ for any starting $\bar{z}_1$ almost equal to $z_1$ because
$z_n = \bar{z}_n$. From same states we can obviously take the same steps to the same successor states.
Furthermore, we can easily take $(\delta/n)$-wide transitions by similar arguments as in the proof
of Lemma~\ref{lem:separation}.
\end{proof}

\subsubsection{Proof of Lemma~\ref{lem:corner}}

\begin{reflemma}{lem:corner}
Let $R$ be a region. For each $\delta > 0$ there is $\delta' > 0$, $n \in \Nset$ and
$z' \in \Gamma_\product$ such that for every $\delta$-separated $z_1 \in R$ there is a $\delta'$-wide path $z_1 \cdots z_n$ such that
$z_n$ and $\bar{z}$ almost equal.

\noindent
Moreover, we can set $n = \maxb + 1$ and $\delta' = \delta / (\maxb+2)$.
\end{reflemma}

\begin{proof}
No relevant clock has in $z_1$ its fractional value in the interval $(0,\delta)$
because $z_1$ is $\delta$-separated. We divide this interval into $\maxb+2$ subintervals
of equal length and set $\delta' = \delta / (\maxb+2)$.

For a fixed $z_1$ we inductively build a $\delta'$-wide path $z_1 \cdots z_n$ where $n=\maxb+1$.
We fix an aribtrary linear order over the set of control states $S$ of the semi-Markov process.
Let $1\leq i < n$. For the state $z_i = (s_i,q_i,\nu_i)$ we choose as $s_{i+1}$ the first state
(in the fixed order) such that $\prob(s_i)(s_{i+1}) > 0$.
This gives us a delay function $f = \delay(s_i,s_{i+1})$. We set $b$ to the integral upper bound
of the interval where $f$ is positive if it is not infinity. Otherwise, we set $b = l+1$ where $l$ is the
lower bound of $f$. Now, we fix the waiting time $t_i = b-\delta'$ and the state $z_{i+1} = (s_{i+1},q_{i+1},\nu_{i+1})$
such that $\trans{z_i}{t_i}{z_{i+1}}$.

We show that it is a $\delta'$-wide transition. We divide the
set of clocks into two disjunct subsets: the set of clocks $\calY$ that have
been reset in one of the states $z_1,\ldots,z_i$ (have been reset at the beginning of the transition
to the next state), and all other clocks $\bar{\calY} = \calX \setminus \calY$.
For each $x \in \calY$ lastly reset in state $z_j$ where $j\leq i$ we have
$\fr(\nu_{i+1}(x)) = 1 - (i+1-j) \cdot \delta'$, i.e. $\fr(\nu_{i+1}(x)) \leq 1 - \delta'$ and
$\fr(\nu_{i+1}(x)) > 1-\delta > \delta'$. For each $x \in \bar{\calY}$ we have
$\fr(\nu_{i+1}(x)) = \fr(\nu_1(x)) - i \cdot \delta' \geq \delta - i \cdot \delta' \geq \delta'$.
Also, $\fr(\nu_{i+1}(x)) \leq 1 - \delta - i \cdot \delta' < 1 - \delta'$.

We show that for any $\delta$-separated starting state $\bar{z}_1 \in R$
we reach a state $\bar{z}_n$ almost equal to $z_n$.
We need a parametrized version of almost equality.
For a set of clocks $\calY$ and two states $z = (s,q,\nu)$ and
$\bar{z} = (\bar{s},\bar{q},\bar{\nu})$ we say that they are $\calY$-equal if
$z \region \bar{z}$ and for each $x\in\calY$ we have $\nu(x) = \bar{\nu}(x)$.
The states $z_1$ and $\bar{z}_1$ are $\emptyset$-equal.
Let $\calX_i$ be a set of clocks and $z_i$, $\bar{z}_i$ be $\calX_i$-equal states.
According to the inductive definition, we fix control states $s_{i+1}, \bar{s}_{i+1}$, waiting times $t,\bar{t}$,
and states $z_{i+1}$ and $\bar{z}_{i+1}$ such that $\trans{z_i}{t}{z_{i+1}}$ and
$\trans{\bar{z}_i}{\bar{t}}{\bar{z}_{i+1}}$. Notice that $s_{i+1}=\bar{s}_{i+1}$, hence $t=\bar{t}$.
We have $z_{i+1} \region \bar{z}_{i+1}$.
Furthermore, they are $(\calX_i \cup \calY)$-equal where $\calY$ is the set of clocks reset in $z_i$.
We get that $z_n$ and $\bar{z}_n$ are almost equal because the paths take at least
$\maxb + 1 \cdot (1-\delta') > \maxb$ time units. All clocks not reset during this path become irrelevant.
We finish the proof by setting $z' = z_n$.
\end{proof}

\section{Proofs of Section~\ref{SEC:ALGORITHMS}}

\subsection{Proof of Theorem~\ref{thm:reachability}}

\begin{reftheorem}{thm:reachability}
\tmreachability
\end{reftheorem}
As $\probm_\smp(\setofruns)$ is equal to the probability of reaching
$\calC$, the $i$-step transition
probabilities $\kernel^i_\product(z,\calC)$ converge to
$\probm_\smp(\setofruns)$ as $i$ goes to infinity. Our goal is to show
that they converge exponentially quickly.

Our proof proceeds as follows. Denote by $B$ the union of all regions
that belong to BSCCs of $\regiongraph_\product$.  We show that for $c
= 4 \cdot |V|$ there is a lower bound $p_{bound}>0$ on the probability
of reaching $B$ in at most $c$ steps from any state $z \in
\Gamma_\product$. Note that then the probability of not hitting $B$
after $i= m\cdot c$ steps is at most $(1-p_{bound})^m$. However, this
means that $\kernel^i_\product(z,\calC)$ cannot differ from the
probability of reaching $\calC$ (and thus also from
$\probm_\smp(\setofruns)$) by more than $(1-p_{bound})^m$ because
$\calC\subseteq B$ and the probability of reaching $\calC$ from
$B\smallsetminus \calC$ is $0$.  Moreover, we show that $p_{bound}$ can be
set to $(\pmin \cdot \densb \cdot 1/c)^{c}$, from which we obtain the desired upper bound on
$|\probm_\smp(\setofruns)-\kernel^i_\product(z,A)|$.

So to obtain the desired result, it suffices to prove the following

\begin{proposition}
\label{prop:reachability}
For every $z\in\Gamma_\product$ we have that 
\[
\kernel^{c}_\product(z,B)\quad \geq\quad p_{bound}
\]
Here
$c = 4 \cdot |V|$ and $p_{\mathit{bound}} = (\pmin \cdot \densb
\cdot 1/c)^{c}$.
\end{proposition}

Note that this section draws heavily on some of the methods and lemmas proved in the
previous section, though often in a slightly easier form. However, to keep
individual sections of the Appendix independent, we repeat the arguments
here once more.

Similarly to previous section, we are interested in paths $z\ldots z_n$  that
are \emph{$\delta$-wide}.  For a fixed $\delta>0$, it means that the waiting
time of any transition in the path can be changed by $\pm\delta$ without ending
up in a different region in the end. Precise definition follows.

\begin{definition}

Let $z = (s,q,\nu)$ and $z' = (s',q',\nu')$ be two states. For a waiting time $t
\in \Rsetp$ we set $\trans{z}{t}{z'}$ if $\successor(z) = (q',\bar{\nu})$ and
$\nu' = \bar{\nu} + t$.  We set $\trans{z}{}{z'}$, called a \emph{feasible
transition}, if  for some $t\in\Rsetp$ (i) $\trans{z}{t}{z'}$; and (ii) $f_d(t) > 0$, where
$f_d = \dist(s,s')$.

For $\delta>0$, we say that a feasible transition $\trans{z}{}{z'}$ is
$\delta$-wide if for every $x \in \calX$ relevant for $\nu'$ we have
$\fr(\nu_i(x)) \in [\delta,1-\delta]$.

Let $z_1 \cdots z_n$ be a path.  It is \emph{feasible} if for each $1 \leq i <
n$ we have that $\trans{z_i}{}{z_{i+1}}$.  It is \emph{$\delta$-wide} if for
each $1 \leq i < n$ we have that $\trans{z_i}{}{z_{i+1}}$ is a $\delta$-wide
transition.  

\end{definition}

We first show that any $\delta$-wide path of a finite length, say $n$, from any state
$z\in \Gamma_\product$ to a state $z_n$ in a region $R$, induces a set of paths
from $z$ to the region $R$, and that their probability is bounded below by a
positive constant.

\begin{lemma}\label{lem:reach-fuzzying}
For every $\delta > 0$ and $n>1$ there is $b >0$ such that the following holds.
For every $\delta$-wide path
$\sigma = z_0z_1 \cdots z_n$, there is a set of states
$Z \ni z_n$ such that it holds $\kernel_\product^n(z_1,Z) \geq b$.

Moreover, we can set $b=(\pmin \cdot \densb \cdot 2\delta/n)^n$.
\end{lemma}

\begin{proof}
We fix any $\delta$-wide path $\sigma=z_0z_1\cdots
z_n=(s_0,q_0,\nu_0)(s_1,q_1,\nu_1)\cdots(s_n,q_n,\nu_n)$. For $1\leq
i\leq n$, let $t_i$ be the waiting times such that
$\trans{z_{i-1}}{t_i}{z_{i}}$, and let $X_i=\{x\in\calX\mid
\successor(z_{i-1}) = (q,\nu),\; \nu(x)=0\}$ be the set of clocks reset right before waiting $t_i$.

For $\varepsilon>0$, we define an $\varepsilon$-neighbourhood of
$\sigma$ to be the set of paths of the form
$\trans{z_0}{t'_1}(s_1,q_1,\nu'_1)\cdots\trans{}{t'_n}{(s_n,q_n,\nu'_n)}$
where $t'_i\in (t_i-\varepsilon,t+\varepsilon)$. Due to
$\delta$-wideness of $\sigma$, all paths of its
$\delta/n$-neighbourhood are feasible, and follow the same sequence of regions. Considering this
$\delta/n$-neighbourhood, the set of all possible $\nu'_n$s forms
the sought set of states $Z$.

We now give a lower bound on $\kernel^n_\product(z,Z)$. First, recall the following
notation:  let $\pmin$ denote the smallest probability in $\smp$. Further, let
us denote by $\delay(\smp)$ the set of delay densities used in $\smp$,
i.e.~$\delay(\smp)=\{\D(s,s')\mid s,s'\in S\}$.  From our assumptions imposed
on delay densities we obtain the following uniform bound $\densb>0$ on delay
densities of $\delay(\smp)$. For every $f\in\delay(\smp)$ and for all
$x\in[0,\maxb]$, either $f(x)>\densb$ or $f(x)=0$, and moreover,
$\int_{\maxb}^\infty f(x)dx>c$ or equals $0$.

We define sets of states $Z_0, Z_1 \dots, Z_n=Z$, where $Z_i$ is the set of all
states $(s_i,q_i,\nu'_i)$ in the $\delta$-neighbourhood of $\sigma$.
Note that $\kernel_\product(z_0, Z_1) = \prob(s_0)(s_1) \cdot
\int_{t_0-\delta/n}^{t_0+\delta/n} f_d(t) dt$, where $f_d$ is the appropriate
delay density for this transition. Using the bounds given above, $\kernel_\product(z_0, Z_1) \geq
p_{min} \cdot \int_{t_0-\delta/n}^{t_0+\delta/n} \densb dt = p_{min} \cdot
\densb \cdot 2\delta/n$. Similarly, for any $z'_i \in Z_i$,
$\kernel_\product(z'_i, Z_{i+1}) \geq  p_{min} \cdot \densb \cdot 2\delta/n$
holds by the same arguments. Therefore, from the definition of the n-step
transition kernel, $\kernel_\product^n(z_0,Z) \geq  (p_{min} \cdot \densb \cdot
2\delta/n)^n$.
\end{proof}

We now prove that from any state $z \in \Gamma_\product$, some BSCC reachable from $z$ in
the region graph is also reachable from $z$ along a $\delta$-wide path, and that
this path length is bounded from above by a constant.

We use two steps: first, we show that, from any $z \in \Gamma_\product$,
we can reach a \emph{$\delta$-separated} state along $\delta'$-wide path of bounded
length; second, once
in a $\delta$-separated state, we construct a $\delta''$-wide path of length at
most $|V|$ ending in the BSCC.

\begin{definition}\label{def:reach-delta-separation}
  Let $\delta > 0$. We say that a set $X \subseteq \Rsetpo$ is
  $\delta$-separated if for every $x,y\in X$ either $\fr(x) =
  \fr(y)$ or $|\fr(x) - \fr(y)| > \delta$. 

Further, we say that $(s,q,\nu)\in\Gamma_\product$ is
$\delta$-separated if the set
  $$
    \{ 0 \} \cup \{ \nu(x) \mid x \in \calX, x \textrm{ is relevant
for }\nu \}
  $$
  is $\delta$-separated.
\end{definition}

\begin{lemma}\label{lem:reach-separation}
There is $\delta > 0$ and $n \in \Nset$ such that for any $z_1 \in
\Gamma_\product$ there is a $\delta$-wide path $z_1 \cdots z_n$ such
that $z_n$ is $\delta$-separated.

\noindent
Moreover, we can set $n = \maxb \cdot (|\calX| + 2)$ and $\delta = 1 / ( 2 (|\calX| + 2))$.
\end{lemma}

\begin{proof}
(Same as Lemma~\ref{lem:separation}) To simplify the argumentation we introduce a notion of a \emph{$r$-grid} that marks $r$ distinguished
points (called \emph{lines}) on the $[0,1]$ line segment. In the proof we show that we can place
fractional values of all relevant clocks on such distinguished points.
Let $r \in \Nset$.
We say that a set of clocks $\calY \subseteq \calX$ is \emph{on $r$-grid} in $z$ if
for every $x \in \calY$ relevant in $z$ we have $\fr(\nu(x)) = n/r$ for some $0 \leq n < r$.
For $0 \leq n < r$, we say that the \emph{$n$-th line of the $r$-grid is free} in $z$ if there
is no relevant clock in the $1/2k$-neighborhood of the $n$-th line,
i.e. for any relevant $x\in\calX$ we have $\fr(\nu(x)) \not\in (n/r - 1/2r,n/r + 1/2r)$.

Let $r = |X|+2$. We inductively build a $1/2r$-wide path $z_1\cdots z_n$ where $n = \maxb \cdot r$.
The set $\emptyset$ is on $r$-grid in $z_1$. We show that if a set $\calY_i$
is on $r$-grid in state $z_i$, there is
a $1/2k$-wide transition to $z_{i+1}$ such that $(\calY_i \cup \mathcal{Z})$ is on $r$-grid in $z_{i+1}$
where $\mathcal{Z}$ is the set of clocks newly reset in $z_i$.
There are $|X|+2$ lines on the grid and only $|X|$ clocks. At least two of these lines must be free.
Let $j \neq 0$ be such a line. Let $t$ be a waiting time and $z_{i+1}$ a state
such that $\fr(t) = 1 - j/r$ and $\trans{z_i}{t}{z_{i+1}}$. Such waiting time must be indeed possible
because the interval where the density function of any transition is positive has integral bounds.
The transition $\trans{z_i}{t}{z_{i+1}}$ is $1/2r$-wide because the line $j$ is free in $z_i$.
Furthermore, the set $(\calY_i \cup \mathcal{Z})$ is on $r$-grid in $z_{i+1}$ because the fractional value of
each clock that was previously on $r$-grid was changed by a multiple of $1/r$. The newly reset clocks have
fractional value $1 - j/r$ which is again a multiple of $1/r$.

Next, we show that $\calX$ is on $r$-grid in $z_n$. Clocks reset in this path on $r$-grid in $z_n$.
The remaining clocks are all irrelevant because the path of $\maxb \cdot r$ steps takes at least $\maxb$ time units.
Indeed, each transition in this path takes at least $1/r$ time unit.
According to the definition, $\calX$ is on $r$-grid in $z_n$.
Hence, the state $z_n$ is $1/r$-separated because the distance between two adjacent grid lines is $1/r$.
By setting $\delta = 1/2r$ we get the result.
\end{proof}

\begin{lemma}\label{lem:reach-path}
Let $\delta, \delta' > 0$ and $R$ be a region. Then
there is $n \in \Nset$ such that for every $\delta$-separated $z \in \Gamma_\product$ it holds that
if there is a feasible path from z to z', for a $z'$ in the region $R$, then there is also $i
\leq n$ and a $\delta'$-wide path $z \cdots z_i$ such that $z_i \in \Gamma_\product \cap
R$ is $\delta'$-separated.

Moreover, we can set $n=|V|$ and $\delta' = \delta/|V|$.
\end{lemma}

\begin{proof}
For simplicity, we first transform this path into a $\delta/2^n$-wide one. We
then show how to improve the result to $\delta/n$-wideness.

Let us fix any $\delta$-separated state $z \in \Gamma_\product$, belonging to a particular region,
say $R_s$. We will show that for any region $R_x$ such that  $(R_x,R_s) \in E$ in the region graph, we can find a waiting time $t$ and
$\delta$-separated state $z_1$ belonging to $R_x$, such that $\trans{z}{t}{z_1}$.

As $R_x$ is reachable from $R_s$ in one step in the region graph, there is an interval of
waiting times $(a,b)$ such that for every $t'\in(a,b)$ $\trans{z}{t}{z_1'}$ for
some $z_1'$ from $R_x$.  Moreover, due to $\delta$-separation of $z$, we obtain
$b-a\geq \delta$. Therefore, we can choose the waiting time
$t=(a+b)/2$ and $z_1'$ is $\delta/2$-separated.
Intuitively, we need to `lower' the $\delta$-separation and wideness in each step as we might be
forced to reset a clock, say $x_r$, to a place between two other clocks, say
$x_1, x_2$, with $|\fr(x_1-x_2)| = \delta$. 

Note that if the state $z'$ is reachable from $z$ along a feasible path, it must
be also reachable in at most $|V|$ steps in the region graph. In such
case, we can put $n = |V|$ and the $\delta'$ would be equal to
$\delta/2^{n}$. However, due to $\delta$-separation, for every $x,y\in\calX$
there are at least $n$ values between $\fr(\nu(x))$ and $\fr(\nu(y))$ such that
even if all were fractional values of other clocks, the state would be
$\delta/n$-separated.  Also note that as the path is only $n$ steps long, there
can be at most $n$ different clocks set between any two clocks. Since we know
their ordering in advance, these $n$ different positions are sufficient, and we
can set $\delta' = \delta/|V|$.  \end{proof}

Now we are ready to prove the Proposition~\ref{prop:reachability}.

\begin{proof}[of Proposition~\ref{prop:reachability}] 

Lemma~\ref{lem:reach-separation} together with Lemma~\ref{lem:reach-path} give us an
upper bound on the number of steps $c_b$ needed to hit a state in one of the BSCCs
along a $\delta$-wide path from any state in $\Gamma_\product$: we can set $c_b
=\maxb \cdot (|\calX| + 2) + |V|$ and $\delta= (1 / ( 2 \cdot (|\calX|+2))$. From 
Lemma~\ref{lem:reach-fuzzying} we have
\begin{align*}
 \kernel^{c_b}_\product(z,B)\quad & \geq\quad \left( \frac{\pmin \cdot \densb}{ 2 ( |\calX|+2) \cdot c_b}  \right)^{c_b} \\
\intertext{As  $c_b \leq 2 \cdot |V|$ for all but very small region graphs we have }
 & \geq\quad \left( \frac{\pmin \cdot \densb}{ 2 ( |\calX|+2) \cdot 2 \cdot |V|}  \right)^{2 \cdot |V|} \\
 & \geq\quad \left( \frac{\pmin \cdot \densb}{ (4 \cdot |V|)^2}  \right)^{2 \cdot |V|} \\
 & \geq\quad \left( \frac{\pmin \cdot \densb}{ 4 \cdot |V|}  \right)^{4 \cdot |V|} \\
\intertext{From this, we get the desired}
\kernel^{c}_\product(z,B)\quad & \geq\quad \left( \frac{\pmin \cdot \densb}{c}  \right)^{c}
\end{align*}
where $c = 4 \cdot |V|$.
\end{proof}

\end{document}